\documentclass{article}
\usepackage{setspace}
\usepackage{braket}
\usepackage{bbm}
\usepackage{tcolorbox}
\usepackage{graphicx}
\usepackage{amsthm}
\usepackage{hyperref}
\hypersetup{
    colorlinks,
    citecolor= orange,
    filecolor= blue,
    linkcolor= blue,
    urlcolor= green
}
\usepackage{multicol}
\usepackage{mathrsfs}
\usepackage{appendix}
\usepackage{amssymb}
\usepackage{relsize}
\usepackage{amsthm}
\usepackage{ amssymb }
\usepackage{amsmath}
\usepackage[utf8]{inputenc}
\usepackage[english]{babel}
\usepackage[a4paper, total={6.6in, 8.5in}]{geometry}
\usepackage{amsmath}

\newtheorem{definition}{Definition}
\newtheorem{definition2}{Lemma}

\newtheorem{definition4}{Theorem}
\newtheorem{definition5}{Remark}

\newtheorem{Co}{Corollary}
\usepackage[T1]{fontenc}
\usepackage{lmodern}

\newtheorem{definition9}{Proposition}
\usepackage{tikz-cd}
\begin{titlepage}

\title{Geometric Measures of Complexity for Open and Closed Quantum Systems }
\author{Alberto Acevedo, Antonio Falcó \\ 
Departamento de Matemáticas, Física y Ciencias Tecnológicas \\
Universidad Cardenal Herrera-CEU, CEU Universities,\\ 
Calle San Bartolome, 55, Alfara del Patriarca 46115, Valencia, Spain}
\date{}
\end{titlepage}
\begin{document}
\singlespacing

\maketitle

\vspace{0.9cm}
\textbf{Abstract:}
The unitary dynamics of quantum systems can be modeled as a trajectory on a Riemannian manifold. This theoretical framework naturally yields a purely geometric interpretation of computational complexity for quantum algorithms, a notion originally developed by Michael Nielsen (Circa, 2007). However, for nonunitary dynamics, it is unclear how one can recover a completely geometric characterization of Nielsen-like geometric complexity. The main obstacle to overcome is that nonunitary dynamics cannot be characterized by Lie groups (which are Riemannian manifolds), as is the case for unitary dynamics. Building on Nielsen's work, we present a definition of geometric complexity for a fairly generic family of quantum channels. These channels are useful for modeling noise in quantum circuits, among other things, and analyze the geometric complexity of these quantum channels. 

\begin{multicols}{2}

\section{Introduction}
\;\;\; In the mid 2000's the foundational work leading to a geometric characterization of the complexity of quantum computation was developed by Nielsen and his collaborators in \cite{NielGeo} \cite{dorth1} \cite{dorth2} among other papers. These ideas were later pushed forward in works such as \cite{brandt} in 2009, and then further pushed to other realms of physics such as Black Holes and the Holographic principal in works such as \cite{suss1} \cite{suss2} \cite{suss3}. To the extent that the authors are aware, interest in this geometric approach to the characterization of the complexity of a quantum computation, in particular those carried out by quantum circuits, has remained stifled since its introduction by Nielsen et al. Furthermore, no notion of geometric complexity has been devised which encapsulates noisy quantum computation; e.g. noisy quantum circuits. Many traditional ways of characterizing complexity of quantum computation exists, such as BQP \cite{Niel}, the class of problems solvable by a quantum computer in polynomial time with bounded error, and QMA \cite{Niel}, the quantum analog of NP but for quantum computers. Other methods of determining the complexity of a quantum computation involve circuit depth for the case of universal quantum computation. The idea of geometric complexity differentiates itself from those notions of complexity just mentioned by explicitly leveraging the underlying geometric structure associated with the dynamics generating the quantum computation of interest in order to define a natural complexity of said computations via geodesics on the respective Riemannian manifolds.

Quantum computation may be  described by quantum circuits or quantum annealers whose dynamics are in turn described by elements of a finite-dimensional Lie group. Lie groups are smooth manifolds that admit a Riemannian metric, or more generally a Finsler metric \cite{NielGeo}, which allows for the study of geodesics on this manifold in the sense described by \cite{NielGeo} \cite{dorth1} \cite{dorth2}. These geodesics may be interpreted as optimal length/time quantum computations; e.g. the shortest circuit connecting some input to some desired output given hardware constraints hich characterize the underlying Riemannian metric. This is done by via the relationship between circuit gate count and geodesics on the Lie Group in question \cite{NielGeo}. The question regarding the type of metric structure that one may use to derive
geodesics and hence study optimal quantum computation is still unexhausted and much is left to be desired; in \cite{NielGeo} they present various approaches including Finsler metrics but we will focus on Riemannian metrics. In this work we will focus on those Riemannian metrics characterized by a Hermitian form, the article will be self-contained and all relevant concepts will be developed as we progress.

\vspace{2mm}

The main goal of this work is to develop/motivate the geometric notions of complexity for quantum computations, and quantum dynamics in general, which Nielsen studied in \cite{NielGeo}. We will not focus on the computational challenges of such geometric measures for complexity but rather focus on developing appropriate measures of geometric complexity and studying their properties while motivating said measures with physical intuition. After introducing geometric complexity, we will present/motivate and generalization of geometric complexity for a generic family of quantum channels (this will be referred to as the channel complexity). Basic properties of the channel complexity will be proven and an example will be studied. Furthermore, a simple example will be presented in the penultimate section in order to show an instance of geometric complexity reduction for  noisy channels. Finally, we will also use the framework developed in this paper to broaden the scope of an algebraic framework for quantum computation and complexity developed in recent times by one of the authors of this paper and his collaborators in \cite{Daniela}.

\section{Geometric Complexity for Quantum Computation}
\label{eqn:sec2}
\;\;\; Consider the Riemannian manifold $(SU(N), \langle\cdot,\cdot\rangle_{\Omega})$, where $SU(N)$ is the Lie group of unitary matrices with determinant one ($N^{2}-1$ dimensional manifold). We equipped with $SU(N)$ the following right-invariant Riemannian metric as follows. Assume that
\begin{equation}
\mathfrak{su}(N) =\mathrm{span}\{\mathbf{\hat{E}}_i: 1 \le i \le N^2-1\}
\end{equation}
Then each $\mathbf{\hat{A}}\in \mathfrak{su}(N)$ can be written as
\begin{equation}
\mathbf{\hat{A}} = \sum_{i=1}^{N^2-1} A_i\mathbf{\hat{E}}_i,\end{equation}
where $A_i \in \mathbb{C}.$ Now, we define
\begin{equation}
\mathbf{Vec}(\mathbf{\hat{A}}) := \begin{pmatrix}
A_1 \\
A_2 \\
\vdots \\
A_{N^2-1}
\end{pmatrix} \in \mathbb{C}^{N^2-1}.
\end{equation}
Given an Hermitian definite positive matrix $\Omega \in \mathbb{C}^{(N^2-1)\times(N^2-1)}$ we introduce the inner product $\langle\cdot,\cdot\rangle_{\Omega}:\mathfrak{su}(N)\times \mathfrak{su}(N)\rightarrow \mathbb{R}$ 
as
\begin{equation}
\label{eqn:metric}
 \big\langle \mathbf{\hat{A}},\mathbf{\hat{B}}\big\rangle_{\Omega}:= \frac{1}{N^{2}-1}\mathbf{Vec}(\mathbf{\hat{A}})^{\dagger}\Omega
 \mathbf{Vec}(\mathbf{\hat{B}}),
\end{equation}
with 
\begin{equation}
\Omega = \begin{pmatrix}l_{1} & & \\ & \ddots & \\ & & l_{N^{2}-1}\end{pmatrix}
\end{equation}
and 
\begin{equation}
\mathbf{Vec}(\mathbf{\hat{A}}) = \begin{pmatrix} A_{1} \\ \vdots\\ A_{N^{2}-1} \end{pmatrix} 
\;\;\;\mathbf{Vec}(\mathbf{\hat{B}}) = \begin{pmatrix} B_{1} \\ \vdots\\ B_{N^{2}-1} \end{pmatrix}.
\end{equation}
Here, the matrix $\Omega$ is diagonal with respect to the basis chosen for $\mathfrak{su}(N)$; without loss of generality we will only consider $l_{i}\geq 1$. The canonical choice is the set of all Pauli strings for example. The vectorization $\mathbf{Vec}(\mathbf{\hat{A}})$ of an arbitrary Hermitian matrix $\mathbf{\hat{A}}$ is built from its decomposition into a linear coordination of elements in $\mathfrak{su}(N)$, e.g. when $\mathbf{\hat{A}}=\sum_{i=1}^{N^{2}-1}A_{i}\mathbf{\hat{E}}_{i}$ ($\mathfrak{su}(N) = span\{\mathbf{\hat{E}}_{i}\}_{i}$) then we write $\mathbf{Vec}(\mathbf{\hat{A}}) = (A_{1}, \cdots, A_{N^{2}-1})^{T}$. The Riemannian metric $\big\langle \cdot, \cdot \rangle_{\Omega}$ is a case of a Hermitian form. When $\Omega$ is the identity matrix, the corresponding metric $\langle\cdot,\cdot\rangle_{\mathbb{I}_{N^{2}-1}}$ is the usual bi-invariant metric that is equipped to $SU(N)$, yielding a homogeneous manifold. The geodesics of $(SU(N),\;\; \langle\cdot,\cdot\rangle_{\mathbb{I}_{N^{2}-1}})$ are known to be characterized by one-parameter unitary groups $e^{-it\mathbf{\hat{H}}}$ where $\mathbf{\hat{H}}\in\mathfrak{su}(N)$. For a non-homogeneous $\Omega$ however, some directions on the tangent space of $SU(N)$ at some point $p$, $T_{p}SU(N)$ will be more stretched out than others. For example, consider the following non-homogeneous right-invariant Riemannian metric, which adds a cost term $q$ to interactions involving more than 2 qubits for a total system of $n$ qubits, meaning that $N=2^{n}$. Namely, 
\small
\begin{equation}
\label{eqn:direnct}
\bar{\Omega} = \begin{pmatrix} 1 & & \\ & \ddots & \\ & & 1\end{pmatrix}_{N^\prime\times N^{\prime}}\bigoplus \begin{pmatrix} q & & \\ & \ddots & \\ & & q\end{pmatrix}_{\bar{N}\times \bar{N}}
\end{equation}
\normalsize
then
\begin{equation}
\label{eqn:qstuff}
 \langle \mathbf{\hat{A}},\mathbf{\hat{B}}\rangle_{\bar{\Omega}}:=\frac{1}{4^n-1}Tr\big\{\mathbf{Vec}(\mathbf{\hat{A}})^{\dagger}\bar{\Omega}\mathbf{Vec}(\mathbf{\hat{B}})\big\}.
\end{equation}
Where $N^{'}$ represents the dimensions associated between interactions involving one and two cubits and $\bar{N} = 4^n-1-N^{\prime}$. 
Here, the first matrix in the direct sum appearing in (\ref{eqn:direnct}) constitutes the subspace corresponding to Pauli strings of one and two-qubit interactions; the second term in the direct sum corresponds to the complimentary subspace. The normalization constant $\frac{1}{4^n-1}$ is conventional but not compulsory. The lack of homogeneity in the metric characterized by $\Omega_{1,2}$ leads to geodesics that are not representable as one-parameter unitary groups. 
\subsection{Carnot-Carath\'eodory metric}
\;\;\; It is known that any sub-Riemannian manifold may be equipped with an intrinsic metric known as the Carnot-Carath\'eodory metric  \cite{grom}; for the case of Riemannian metrics $\langle\cdot,\cdot\rangle_{\Omega}$ defined above, the corresponding Carnot-Carath\'eodory metric is defined as follows.
\begin{equation}
\label{eqn:cara}
\mathscr{D}_{\Omega}(\mathbf{\hat{U}},\mathbf{\hat{V}}) := \inf_{\gamma}\int_{t_{0}}^{t_{f}}\sqrt{\langle \gamma^{'}(s),\gamma^{'}(s)\rangle_{\Omega}}ds
\end{equation}
where the infimum is taken along all smooth curves $\gamma:[t_{0},t_{f}]\rightarrow SU(N)$ such that $\gamma(t_{0}) = \mathbf{\hat{U}}$, $\gamma(t_{f}) = \mathbf{\hat{V}}$. For simplicity's sake, we will let $t_{0} = 0 $ and $t_{f} =1$ from now on.
In \cite{NielGeo}\cite{brandt}, techniques are developed for the computation of geodesics $\gamma(s)$ which optimize the integral in the definition of $\mathscr{D}_{\Omega}(\cdot, \cdot)$. We will generally not be concerned with the deduction of such $\gamma(s)$; rather, we will assume the existence of such $\gamma(s)$ and examine properties of such geodesics. Given an element $\mathbf{\hat{U}}(t)\in SU(N)$ (dependent on a parameter $t\geq 0$, the intrinsic metric $\mathscr{D}_{\Omega}(\cdot, \cdot)$ has been used in \cite{NielGeo} \cite{brandt} to define the so-called geometric quantum computational complexity of the unitary operator $\mathbf{\hat{U}}(t)$, in \cite{suss1} such a measure of complexity is simply referred to as the Geometric Complexity. We formalize this notion as a definition below.
\begin{definition} [\textbf{Geometric Complexity of a} $\mathbf{\hat{U}}(t)\in SU(N)$ ] 
Let $\mathbf{\hat{U}}\in SU(N)$, equipping $SU(N)$ with the Carnot-Carath\'eodory metric (\ref{eqn:cara}) for fixed $\Omega$, we define the Geometric Complexity of $\mathbf{\hat{U}}$, denominated $\mathscr{G}_{\Omega}(\mathbf{\hat{U}})$, as follows
\begin{equation}
\mathscr{G}_{\Omega}(\mathbf{\hat{U}}):= \mathscr{D}_{\Omega}(\mathbb{I}, \mathbf{\hat{U}}) = \inf_{\gamma}\int_{0}^{1}\sqrt{\langle\gamma^{'}(s),\gamma^{'}(s)\rangle_{\Omega}}ds 
\end{equation}
\end{definition}
In the context of quantum computation, this notion of complexity strives to supplant the gate/depth notions of computational complexity where a universal set of quantum gates is presupposed to constitute enough expressibility to effectuate all possible gates, and any quantum computation is carried out as a sequential implementation of individual elements of these universal gates \cite{Niel}; the complexity in such cases is related to the number of gates necessary to achieve the desired quantum computational task given a tolerance of error. The notion of Geometric Complexity replaces the concept of a universal set of quantum gates by introducing a generally non-homogeneous right-invariant Riemannian metric $g_{\Omega}$ as defined above. Thenon-homogeneous nature of said metric restricts the paths along which an initial element $\mathbf{\hat{U}} \in SU(N)$ may evolve when connecting from $\mathbf{\hat{U}}$ to some desired final $\mathbf{\hat{V}}\in SU(N)$. 

Let $\mathbf{\hat{U}}$ be a unitary operator, and let us assume that the curve $\gamma(s)$ is a parametrized geodesic curve minimizing the intrinsic distance $\mathscr{D}_{\Omega} (\mathbb{I},\mathbf{\hat{U}})$ (i.e. a geodesic). For all $s\in[0,1]$, the derivative $\gamma^{'}(s) = \mathbf{\hat{H}}(s)\gamma(s)$, where $\mathbf{\hat{H}}(s)$ is an element of $TSU_{\gamma(s)}(N)=\mathfrak{su}(N)$; i.e. the optimal direction of travel at point $\gamma(s)$ along the trajectory. Indeed, the non-homogeneity of the metric $g_{\Omega}$ structures the elements of the tangent space $\mathbf{\hat{H}}(s)$ for each $\gamma(s)$. Assuming we have deduced the family of tangential elements $\mathbf{\hat{H}}(s)$ characterizing the optimal curve $\gamma(s)$, we can naturally reinterpret the $\gamma(s)$ as the solution to the time-dependent Schr\"{o}dinger equation $\gamma^{'}(s)= -i\mathbf{\hat{H}}(s)\gamma(s)$, the solution to said ODE being the following. 
\begin{equation}
\label{eqn:trajj}
\gamma(s) = \hat{\mathscr{T}}e^{-i\int_{0}^{1}\mathbf{\hat{H}}(s)ds}
\end{equation}
With $\gamma(s)$ the optimizer of (\ref{eqn:cara}), and $\mathscr{\hat{T}}$ the time-ordering operator. The corresponding Geometric complexity $\mathscr{G}_{\Omega}(\mathbf{\hat{U}})$ takes the following form. 
\begin{equation}
\mathscr{G}_{\Omega}(\mathbf{\hat{U}})=\frac{1}{\sqrt{N^{2}-1}}\int_{0}^{1}\sqrt{\langle\mathbf{\hat{H}}(s), \mathbf{\hat{H}}(s)\rangle_{\Omega}}ds
\end{equation}
\subsection{Cohering Power and other measures related to the geometric complexity}
\;\;\; There exists a slew of quantum resources, e.g. quantum coherence, entanglement, circuit magic and circuit sensitivity, just to name a few. Given a quantum circuit characterized by a unitary matrix $\mathbf{\hat{U}}$, the corresponding cohering power (magic power, circuit sensitivity, etc) may be defined as optimization functionals mapping unitary matrices to positive real numbers $\max_{x\in D}F(\mathbf{\hat{U}},x)$. Here the optimization is written in a rather abstract form, but as we shall see, the space $D$ will tend to be the space of density matrices. In the cases of circuit sensitivity, cohering power, and magic power, optimization functionals may be defined whose rate of change is governed by the generator of $\mathbf{\hat{U}}$ \cite{Bu}; this ultimately leads to lower bounds on the geometric complexity of some $\mathbf{\hat{U}}\in SU(N)$ which imply rigorously that quantum resources are inherently related to the geometric complexity of the corresponding quantum circuit. For completeness, we present below the main result of \cite{Bu}. 
\begin{definition4}{Lower bounds on circuit complexity} The circuit cost of a quantum circuit $\mathbf{\hat{U}}\in SU(2^{n})$ is lower bounded as follows:
\begin{equation}
Cost(\mathbf{\hat{U}})
 \geq c \max\Big\{CiS(\mathbf{\hat{U}}), \frac{\mathscr{M}(\mathbf{\hat{U}})}{4}, \frac{\mathscr{C}_{r}(\mathbf{\hat{U}})}{log(2)}\Big\}
\end{equation}
where $c$ is a universal constant independent of $n$. The quantities $CiS(\mathbf{\hat{U}})$, $ \mathscr{M}(\mathbf{\hat{U}})$ and $\mathscr{C}_{r}(\mathbf{\hat{U}})$ quantify the sensitivity, magic and coherence respectively, see \cite{Bu}.
\end{definition4}
In \cite{Bu} rather than use $\mathscr{G}_{\Omega}(\mathbf{\hat{U}})$ they define the geometric complexity as follows (referring to it as the cost)
\begin{equation}
Cost(\mathbf{\hat{U}}):=\inf_{\gamma(s)}\int_{0}^{1}\sum_{j=1}^{m}|r_{j}(s)|ds
\end{equation}
where the $r_{j}(s)$ are the coefficients if the expansion of $\mathbf{\hat{H}}(s)$ or time $s$; where the $\mathbf{\hat{H}}(s) $ constitute a trajectory of the likes of (\ref{eqn:trajj}). With the latter in mind, let us now relate our definition of geometric complexity to the definition of cost in \cite{Bu}. 
\begin{equation}
Cost(\mathbf{\hat{U}}):=\inf_{\gamma(s)}\int_{0}^{1}\sum_{j=1}^{m}|r_{j}(s)|ds \leq
\end{equation}
\begin{equation}
\inf_{\gamma(s)}\int_{0}^{1}m(\gamma(s))\sqrt{\sum_{j=1}^{m(\gamma(s))}|r_{j}(s)|^{2}}ds = 
\end{equation}
\begin{equation}
\inf_{\gamma(s)}\int_{0}^{1}m(\gamma(s))\sqrt{\langle \mathbf{\hat{H}}(s), \mathbf{\hat{H}}(s)\rangle_{\Omega}}ds \leq  
\end{equation}
\begin{equation}
(N^{2}-1)\mathscr{G}_{\Omega}(\mathbf{\hat{U}}) < N^{2}\mathscr{G}_{\Omega}(\mathbf{\hat{U}})
\end{equation}
The latter yields the following corollary. 
\begin{Co}{Lower bounds on geoemtric complexity part 2}
\begin{equation}
\mathscr{G}_{\Omega}(\mathbf{\hat{U}})
 \geq c \max\Big\{\frac{CiS(\mathbf{\hat{U}})}{4^n}, \frac{\mathscr{M}(\mathbf{\hat{U}})}{4^{n}}, \frac{\mathscr{C}_{r}(\mathbf{\hat{U}})}{log(2)4^n}\Big\}
\end{equation}
where $c$ is a universal constant independent of $n$.
\end{Co}

We will now center our attention on the notion of cohering power. We will now diverge from the calcuations of \cite{Bu} and derive our own lower bound for geometric complexity from coherence power using a different approach to the one implemented in \cite{Bu}. To formally define the notion of cohering power, we first introduce the concept of relative entropy of coherence with respect to the purity/linear entropy. In \cite{Bu} the von Neumann entropy is used, here we opt out for the simpler linear entropy. 
\begin{definition}[\textbf{Relative entropy of coherence}]
Let $\mathscr{E}(\cdot):\sum_{i}\mathbf{\hat{P}}_{i}(\cdot)\mathbf{\hat{P}}_{i}$ be a completely-dephasing channel, where $\{\mathbf{\hat{P}}_{i}\}_{i}$ is an orthogonal set of projectors that spans $\mathbb{C}^{N}$. Now, let $\boldsymbol{\hat{\rho}}\in S(\mathbb{C}^{N})$ ( the space of density operators acting in $\mathbb{C}^{N}$). The Relative Entropy of Coherence is defined as follows. 

\begin{equation}
C_{\mathscr{E}}\big( \boldsymbol{\hat{\rho}}\big):= \zeta(\boldsymbol{\hat{\rho}})-\zeta(\mathscr{E}(\boldsymbol{\hat{\rho}}))
\end{equation}
where 
\begin{equation}
\zeta(\mathbf{\hat{A}}) := Tr\big\{\mathbf{\hat{A}}^{2}\big\}
\end{equation}
known as the purity of a quantum state. 
\end{definition} 
Notice that $C_{\mathscr{E}}$ may be defined equivalently in terms of the linear entropy $S_{L}\big(\mathbf{\hat{\rho}}\big):=1-\zeta(\mathbf{\hat{\rho}})$ with respect to the completely-dephasing map $\mathscr{E}$ as follows. 
\begin{equation}
C_{\mathscr{E}}\big( \boldsymbol{\hat{\rho}}\big):= S_{L}\big(\mathscr{E}\big(\boldsymbol{\hat{\rho}}\big)\big)-S_{L}\big(\boldsymbol{\hat{\rho}}\big)
\end{equation}
As mentioned already, in \cite{Bu} the von Neumann entropy has been used in lieu of the linear entropy which we employ here. We have chosen to utilize the linear entropy/ purity because it is simpler to manipulate, and it serves one of our ultimate goals in this work, which is to bound geometric complexity for quantum computation and to relate it to show the relationship between geometric complexity and quantum resources such as coherence. To proceed, we define our version of the notion of cohering power, for a given unitary operator $\mathbf{\hat{U}}\in SU(N)$, introduced in \cite{Bu}. 
\begin{definition}[\textbf{Cohering Power of a unitary evolution} $\mathbf{\hat{U}}$]
\label{eqn:coheringpower}
We fix a completely-dephasing channel $\mathscr{E}$, as in the definition of Relative entropy coherence. Here, the Cohering Power of $\mathbf{\hat{U}}$ with respect to $\mathscr{E}$ 
 for a fixed $t\geq 0$ is defined as follows. 
\begin{equation}
\label{eqn:decpower}
\mathscr{C}_{\mathscr{E}}(\mathbf{\hat{U}}):=\max_{\boldsymbol{\hat{\rho}}\in S(\mathbb{C}^{N})}\big|C_{\mathscr{E}}(\mathbf{\hat{U}}\boldsymbol{\hat{\rho}}\mathbf{\hat{U}}^{\dagger})-C_{\mathscr{E}}(\boldsymbol{\hat{\rho}})\big|
\end{equation}
where the maximum is taken over all density operators acting in $\mathbb{C}^{N}$, which we denote by $S(\mathbb{C})$; for $n$ qubits $N= 2^{n}$.
\end{definition}
Furthermore, we introduce the related notion of "Rate of Coherence". Our treatment of this quantity deviates from the approach found in \cite{Bu}.
\begin{definition}[\textbf{Rate of Coherence}]
Given an s-dependent Hermitian matrix $\mathbf{\hat{H}}(ts)$ in the span of $\mathfrak{su}(N)$, define the rate of coherence of the unitary evolution $\mathbf{\hat{U}}_{t}:=e^{-i\int_{0}^{1}\mathbf{\hat{H}}(ts)ds} = e^{-i\int_{0}^{t}\mathbf{\hat{H}}(s)ds}  $ with respect to the state $\boldsymbol{\hat{\rho}}$ and the completely-dephasing map $\mathscr{E}$ at $t=t^{'}$ as follows. 

\begin{equation}
R_{\mathscr{E}}(\mathbf{\hat{H}}(t),\boldsymbol{\hat{\rho}})\Big|_{t=t^{'}}:= \frac{d}{dt}C_{\mathscr{E}}\big(\mathbf{\hat{U}}_{t}\boldsymbol{\hat{\rho}}\mathbf{\hat{U}}^{\dagger}_{t}\big)\Big|_{t=t^\prime}
\end{equation}
\end{definition}

\begin{definition9}[\textbf{Exact value of the rate of change and a bound}]
\label{eqn:ratethe}
Let us first make the following two definitions.
\begin{equation}
\mathbf{\hat{U}}_{t}:=e^{-i\int_{0}^{t}\mathbf{\hat{H}}(s)ds}
\end{equation}
and 
\begin{equation}
\boldsymbol{\hat{\rho}}_{t}:=\mathbf{\hat{U}}_{t}\boldsymbol{\hat{\rho}}\mathbf{\hat{U}}^{\dagger}_{t}
\end{equation}
\label{eqn:prop1}
Then,
\begin{equation}
R_{\mathscr{E}}(\mathbf{\hat{H}}(t),\boldsymbol{\hat{\rho}})\Big|_{t=t^\prime} = 2iTr\big\{\big[\boldsymbol{\hat{\rho}}_{t^\prime},\mathscr{E}\big(\boldsymbol{\hat{\rho}}_{t^{\prime}}\big)\big]\mathbf{\hat{H}}(t^\prime)\big\}
\end{equation}
and
\begin{equation}
\big|R_{\mathscr{E}}(\mathbf{\hat{H}}(t),\boldsymbol{\hat{\rho}})\big|\Big|_{t=t^\prime}\leq A_{\mathscr{E},\boldsymbol{\hat{\rho}}}\big\|\mathbf{\hat{H}}(t^\prime)\big\|_{H.S.}
\end{equation}
where $A_{\mathscr{E},\boldsymbol{\hat{\rho}}} := \sup_{t^\prime\in[0,t]}\big\|\big[\boldsymbol{\hat{\rho}}_{t^\prime},\mathscr{E}\big(\boldsymbol{\hat{\rho}}_{t^{\prime}}\big)\big]\big\|_{H.S.}\leq \sqrt{2}$
\end{definition9}
\begin{proof}
See Appendix \ref{eqn:ratethemdem}
\end{proof}

Once again letting $\mathbf{\hat{U}}_{t}:=e^{-i\int_{0}^{t}\mathbf{\hat{H}}(s)ds}$, we can use the fundamental theorem of calculus to relate the cohering power to the rate of coherence as follows. 
\begin{equation}
|C_{\mathscr{E}}(\mathbf{\hat{U}}_{t}\boldsymbol{\hat{\rho}}\mathbf{\hat{U}}^{\dagger}_{t})-C_{\mathscr{E}}(\boldsymbol{\hat{\rho}})| = 
\end{equation}
\begin{equation}
\Big|\int_{0}^{t} \frac{d}{ds}C_{\mathscr{E}}\big(\mathbf{\hat{U}}_{s}\boldsymbol{\hat{\rho}}\mathbf{\hat{U}}^{\dagger}_{s}\big)\Big|_{s=t^\prime}dt^\prime\Big| = 
\end{equation}
\begin{equation}
\Bigg|\int_{0}^{t}R_{\mathscr{E}}(\mathbf{\hat{H}}(t^\prime),\boldsymbol{\hat{\rho}})dt^\prime\Bigg| 
\end{equation}
where $\boldsymbol{\hat{\rho}}$ is taken to be the the maximizor of (\ref{eqn:decpower}) for time $t$.
The latter leads us to the following interesting result relating Geometric Complexity to Decoherence power.
\begin{definition4}[\textbf{Decohering power bounds the geometric complexity from below}]
\label{eqn:devoheringpowerbound}
\begin{equation}
\frac{1}{\sqrt{2N}}\mathscr{C}_{\mathscr{E}}(\mathbf{\hat{U}}_{t})\leq \mathscr{G}_{\Omega}(\mathbf{\hat{U}}_{t})
\end{equation}
\end{definition4}
\begin{proof}
 \begin{equation}
 \label{eqn:nl1}
\mathscr{C}_{\mathscr{E}}(\mathbf{\hat{U}}_{t})=\Bigg|\int_{0}^{t}R_{\mathscr{E}}(\mathbf{\hat{H}}(s),\boldsymbol{\hat{\rho}})ds\Bigg| \leq \int_{0}^{t}\Big|R_{\mathscr{E}}(\mathbf{\hat{H}}(s),\boldsymbol{\hat{\rho}})\Big|ds \leq
\end{equation}
\begin{equation}
\label{eqn:nl2}
\int_{0}^{t}\sqrt{2}\|\mathbf{\hat{H}}(s)\|_{H.S.}ds \leq
\int_{0}^{t}\sqrt{2}\sqrt{Tr\big\{\mathbf{\hat{H}}(s)\Omega\mathbf{\hat{H}}(s)\big\}}ds =
\end{equation}
\begin{equation}
\int_{0}^{t}\sqrt{2}\frac{\sqrt{N^{2}-1}}{\sqrt{N^{2}-1}}\sqrt{Tr\big\{\mathbf{\hat{H}}(s)\Omega\mathbf{\hat{H}}(s)\big\}}ds = 
\end{equation}
\begin{equation}
\sqrt{2(N^{2}-1)}\int_{0}^{t}\sqrt{\langle \mathbf{\hat{H}}(s), \mathbf{\hat{H}}(s)\rangle_{\Omega}}ds = \sqrt{2}N\mathscr{G}_{\Omega}(\mathbf{\hat{U}}_{t}) 
\end{equation}

where we have used Proposition \ref{eqn:prop1} in going from the end of (\ref{eqn:nl1}) to the beginning of (\ref{eqn:nl2}).
\end{proof}
This result differs from the analogous result in \cite{Bu} in that here we have included the information from the underlying metric $\langle\cdot, \cdot\rangle_{\Omega}$ and there is an overall $\frac{1}{N^{2}-1}$ difference in the lower bound. 

\section{Complexity in Open Quantum Systems}

\;\;\;So far we have studied geometric complexity in relation to unitary dynamics characterizing quantum computation. We now pass on to a generalization of geometric complexity for quantum computation for the case where the interaction between the quantum computer and its environment plays a non-negligible role. This puts us within the domain of open quantum systems; the natural objects of study in this field are quantum channels \cite{Niel}, of which the Markovian type are the most popular example \cite{Schloss}. Informally, a quantum channel $\Lambda(\cdot):\mathcal{S}(\mathscr{H})\rightarrow \mathcal{S}(\mathscr{H})$ is a completely positive map that takes density operators to density operators; here $\mathcal{S}(\mathscr{H})$ represents the space of density operators acting in the Hilbert space $\mathscr{H}$. We must emphasize the use of density operators here since in general, the outputs of quantum channels need not be pure states. This indeed means that the dynamics afforded by the given quantum channel $\Lambda(\cdot)$ will not be unitary. This causes a paradigm shift; where before, with unitary channels, we could study the dynamics on a quantum state as an element of a Lie Group, for the case of non-unitary dynamics, this is no longer the case as the set of quantum channels for a fixed Hilbert space does not constitute a Lie Group. In the best of cases, we can obtain a so-called Lie Semigroup which does not have a smooth differential structure like Lie groups do. An example of such a case are the well-studied Markovian open quantum systems, where we focus on one-parameter semigroups.  

Drawing parallels between the Lie Group description of unitary dynamics and non-unitary dynamics, in a geometric sense, is not trivial since the mathematical structures underlying these so-called Lie semigroups are very different, invoking concepts such as Lie wedge and tangent cone, see \cite{cone}. The power of Lie groups is that they are smooth manifolds and as such may be equipped with a Riemannian metric such as the Hermitian forms $\langle \cdot, \cdot\rangle_{\Omega}$ studied in the previous section. Doing something that parallels the latter would require one to first work within constraints that conduce to a smooth manifold encompassing interesting non-unitary (in some sense) dynamics and equipping said set with a Riemannian metric. To the authors, the latter remains nebulous, and we will therefore look for alternatives. One way to avoid this conundrum is to focus on the geometry of the states rather than that of the unitary operators describing the dynamics. The sets of pure states and mixed states can both be equipped with the Bures/Fubini-Study Riemannian metric, which then grounds us in the realm of Riemannian geometry. We will say more on this later. Another approach, and the one that will be the focus of this work, is to present a way to connect from the non-unitary channel case to the unitary case via a particular purification scheme that is more amenable to the techniques already in play involving Riemannian geometry on Lie Groups. In the following, we will present a family of quantum channels that are generic enough to describe the canonical models of decoherence and dissipation. We will use such channels to model noise effectuated by the environment $E$ onto the system $S$ (The quantum computer). 
\section{Non-unitary quantum channels in Open Quantum Systems}
\;\;\; We will focus on non-unitary channels $\Lambda_{t}(\cdot)$ that may be generated by fixing a Hilbert space $\mathscr{H}= \mathscr{H}_{S}\otimes\mathscr{H}_{E}$, letting $\sum_{i}p_{i}|E_{i}\rangle\langle E_{i}| \in \mathcal{S}(\mathcal{\mathscr{H}}_{E})$ be the reduced state of $E$ at $t=0$ ( we have the freedom to represent the state of the environment in whatever basis we want), and $\boldsymbol{\hat{\rho}}_{S} \in \mathcal{S}(\mathcal{\mathscr{H}}_{S})$. Next, consider the time-independent total Hamiltonian $\boldsymbol{\hat{H}}_{tot} = \boldsymbol{\hat{H}}_{S}+\boldsymbol{\hat{H}}_{I}+\boldsymbol{\hat{H}}_{E}$
where $\boldsymbol{\hat{H}}_{S}$, $\boldsymbol{\hat{H}}_{I}$ and $\boldsymbol{\hat{H}}_{E}$ are respectively the Hamiltonians pertaining to the self-dynamics of the system, the interactive dynamics and the self-dynamics of the environment; indeed, $\mathbf{\hat{H}}_{S}$ acts non-trivially only in $\mathscr{H}_{S}$, $\mathbf{\hat{H}}_{E}$ acts non-trivially only in $\mathscr{H}_{E}$, and $\mathbf{\hat{H}}_{I}$ acts non-trivially in the total space $\mathscr{H}_{S}\otimes \mathscr{H}_{E}$. Starting with some initial separable state $\boldsymbol{\hat{\rho}}_{0} = \boldsymbol{\hat{\rho}}_{S}\otimes\sum_{i}p_{i}|E_{i}\rangle\langle E_{i}|$, we may obtain the local non-unitary dynamics of the system $S$ by taking the partial trace over the environmental degrees pf freedom as follows. 
\begin{equation}
\label{eqn:nonuni2}
\Lambda_{t}(\boldsymbol{\hat{\rho}}_{S}):= 
\end{equation}
\begin{equation}
Tr_{E}\Big\{e^{-it\boldsymbol{\hat{H}}_{tot}}\Big(\boldsymbol{\hat{\rho}}_{S}\otimes\sum_{i}p_{i}|E_{i}\rangle\langle E_{i}|\Big)e^{it\boldsymbol{\hat{H}}_{tot}}  \Big\} = 
\end{equation}
\begin{equation}
=\sum_{ij}p_{i}\langle E_{j}|e^{it\boldsymbol{\hat{H}}_{tot}}|E_{i}\rangle\boldsymbol{\hat{\rho}}_{S}\langle E_{i}|e^{-it\boldsymbol{\hat{H}}_{tot}}|E_{j}\rangle\end{equation}
The operators $\boldsymbol{\hat{M}}_{ji}(t):= \sqrt{p_{i}}\langle E_{j}|e^{it\boldsymbol{\hat{H}}_{tot}}|E_{i}\rangle$ are a case of what are known as Kraus operators \cite{Niel}; the key features of said operators, for the case of quantum channels, is their role in the following resolution of the identity for every $t\geq 0$
\begin{equation}
\sum_{ji}\boldsymbol{\hat{M}}_{ji}^{\dagger}(t)\boldsymbol{\hat{M}}_{ji}(t)= \mathbb{I}_{S}  
\end{equation}
whence we rewrite 
\begin{equation}
\label{eqn:nonuni}
\Lambda_{t}(\boldsymbol{\hat{\rho}}_{S})=\sum_{ij}\boldsymbol{\hat{M}}_{ji}(t)\boldsymbol{\hat{\rho}}_{S}\boldsymbol{\hat{M}}_{ji}^{\dagger}(t).
\end{equation}
Although the quantum channel $\Lambda_{t}(\cdot)$ is defined in a rather generic way, the assumptions leading to the non-unitary dynamics in (\ref{eqn:nonuni}) are nevertheless not the most generic. Further levels of generality would involve relaxing the assumption that the initial state $\boldsymbol{\hat{\rho}}_{tot}$ is a product state. There is not much to gain conceptually from such a generalization since one could pass from a setting where the initial state $\boldsymbol{\hat{\rho}}_{tot}$ is a product state to the case where we consider a linear combination of product states evolving non-unitarily via the linearity of the quantum channel $\Lambda_{t}(\cdot)$ defined above. As a consequence, we henceforth only focus on initial $\boldsymbol{\hat{\rho}}_{tot}$, which are product states. The second level of generality is the possibility of $\boldsymbol{\hat{H}}_{tot}$ to be time-dependent. Here, the shift from a time-independent total Hamiltonian to one that is time-dependent will not be so trivial and will require special attention, leading to some more interesting results. We will hence come back to these later sections; for the sake of motivating the general ideas of this work, we will focus on the time-independent $\mathbf{\hat{H}}_{tot}$ for the moment.

Let us now ask the following questions. "What is the geometric complexity of the time-dependent quantum channel defined by $\Lambda_{t}(\cdot)$"?. "What could such a question mean in physical terms"? In Theorem \ref{eqn:devoheringpowerbound}, we showed that the geometric complexity is bounded from below by a term proportional to the cohering power (Definition \ref{eqn:coheringpower}) of said unitary evolution; hence establishing a relationship between the "complexity" of states and the complexity of the dynamics. Namely, it quantified the geometric complexity of a unitary dynamics via its capacity to generate complex states, where the level of complexity of a quantum state, here may be measured in terms of the von Neuman entropy for which a fully decohered quantum state has max entropy and a pure state (no decoherence) has minimum entropy. We may therefore generalize this physical motivation in terms of unitary dynamics to the case of non-unitary dynamics, i.e., how capable is the quantum channel $\Lambda_{t}(\cdot)$ of generating complex states in the same sense as before. Owing to the fact that non-unitary maps are entropy increasing, we would expect whatever measure of geometric complexity is defined for general quantum channels to recover the usual geometric complexity values for the unitary case while proving inferior measures for the non-unitary case. 

\section{Geometric complexity of a quantum channel}
\;\;\; To define a geometric complexity for quantum channels let us first assume that $\Lambda_{t}(\cdot)$ is defined as in (\ref{eqn:nonuni}). In general, there is not a unique larger Hilbert space, initial state $\sum_{i}p_{i}|E_{i}\rangle\langle E_{i}|$, and unitary dynamics $e^{-it\boldsymbol{\hat{H}}_{tot}}$ yielding the map $\Lambda_{t}(\cdot)$. Nevertheless, in physical settings one knows the initial conditions leading to $\Lambda_{t}(\cdot)$ and so we may simply identify this non-unitary map with $\boldsymbol{\hat{U}}_{t}(\cdot):=e^{-it\boldsymbol{\hat{H}}_{tot}}(\cdot)e^{it\boldsymbol{\hat{H}}_{tot}}$; from now on we will always assume that such initial conditions leading to the map $\Lambda_{t}(\cdot)$ are known and unambiguous. This brings us back to the setting of unitary dynamics. We may therefore compute the geometric complexity of such a dynamics by just considering the unitary operator $e^{-it\boldsymbol{\hat{H}}_{tot}}$ and using the definition of geometric complexity from the previous sections. Since we are only interested in time-independent total Hamiltonians $\mathbf{\hat{H}}_{tot}$, this forces the case where $\Omega$ is just the identity matrix; i.e. all $l_{i} =1$. This is due to the fact that time-independent Hamiltonians are generators of one-parameter unitary groups which are a parametrization of a geodesic on $SU(N)$ with the flat metric, which we will label with an index $hs$ for Hilbert-Schmidt, i.e. the Hilbert-Schmidt metric. Now, let $d:= dim(\mathscr{H}_{S}\otimes\mathscr{H}_{E})$, then
\begin{equation}
\mathscr{G}_{hs}\big(e^{-it\boldsymbol{\hat{H}}_{tot}}\big)= 
\end{equation}
\begin{equation}
\frac{1}{\sqrt{d^{2}-1}}\int_{0}^{t}\sqrt{\langle\boldsymbol{\hat{H}}_{tot}, \boldsymbol{\hat{H}}_{tot}\rangle_{hs}}dt^{\prime} =
\end{equation}
\begin{equation}
\frac{1}{\sqrt{d^{2}-1}}\int_{0}^{t}\big\|\boldsymbol{\hat{H}}_{tot}\big\|_{H.S.}dt^{\prime}= \frac{t}{\sqrt{d^{2}-1}}\big\|\boldsymbol{\hat{H}}_{tot}\big\|_{H.S.}
\end{equation}
However, this gives us way more information than we desire. What we would like is to distill the complexity associated with the dynamics of the system $S$, and so a new definition of geometric complexity is necessary. Before producing such a novel definition, note that here $\|\mathbf{\hat{H}}_{tot}\|_{H.S.} = \||\mathbf{\hat{H}}_{tot}|\|_{H.S.}$, which implies that $\mathscr{G}_{hs}\big(e^{-it\boldsymbol{\hat{H}}_{tot}}\big) = \mathscr{G}_{hs}\big(e^{-it|\boldsymbol{\hat{H}}_{tot}|}\big)$; we state this as a formal remark. 
\begin{definition5}
\label{eqn:remark1}
Consider an arbitrary Hamiltonian $\mathbf{\hat{H}}$, then 
\begin{equation}
\mathscr{G}_{hs}\big(e^{-it\boldsymbol{\hat{H}}}\big) = \mathscr{G}_{hs}\big(e^{-it|\boldsymbol{\hat{H}}|}\big)
\end{equation}
\end{definition5}
Furthermore, the following is also easy to show.
\begin{definition5}
\label{eqn:remark2}
Consider arbitrary Hamiltonians $\mathbf{\hat{H}}_{1}$ and $\mathbf{\hat{H}}_{2}$, then 
\begin{equation}
\mathscr{G}_{hs}\big(e^{-it(\boldsymbol{\hat{H}}_{1}+\mathbf{\hat{H}}_{2})}\big) \leq \mathscr{G}_{hs}\big(e^{-it\boldsymbol{\hat{H}}_{1}}\big)+\mathscr{G}_{hs}(e^{-it\boldsymbol{\hat{H}}_{2}}\big)
\end{equation}
\end{definition5}
With the latter in mind, we present the following generalization of geometric complexity. As mentioned before, we begin with the case of dynamics generated by a one-parameter group and then move on to the more general case in later sections. 
\begin{definition}[Geometric complexity for Quantum Channels generated by one parameter groups, Channel Complexity for short ]
\label{eqn:orbitcomplexity}
Let $\Lambda_{t}$ (we will drop parentheses $(\cdot)$ for clarity), be defined as in (\ref{eqn:nonuni}), characterized by the one-parameter unitary group $e^{-it\boldsymbol{\hat{H}}_{tot}}$. We define a geometric complexity measure, which we will denominate as the Channel-Complexity, as follows. 
\begin{equation}
\mathcal{G}_{hs}\big(\Lambda_{t}\big):=\mathscr{G}_{hs}\big(e^{-it\boldsymbol{\hat{H}}_{tot}}\big) - 
\mathscr{G}_{hs}(e^{-it\big(\sqrt{|\mathbf{\hat{H}}_{tot}^{2}-\mathbf{\hat{H}}_{S}^{2}|}\big)})
\end{equation}
where $\mathscr{G}_{hs}\big(e^{-it\boldsymbol{\hat{H}}_{tot}}\big)$ is the usual geometric complexity defined in (\ref{eqn:cara}). 
\end{definition}
At first glance, it might be evident that this definition's effect is to eliminate information regarding the dynamics that is only detectable in the environmental degrees of freedom. To convince the reader that this is indeed what is achieved in some sense, we will provide evidence that our measure is sensible for the detection of "complexity-loss" due to non-unitary channels. 

To showcase teh utility of Definition \ref{eqn:orbitcomplexity}, we first consider the case where $\mathbf{\hat{H}}_{I} = 0 = \mathbf{\hat{H}}_{E}$. Here, $\mathcal{G}_{hs}(e^{-it\mathbf{\hat{H}_{S}}}) = \mathscr{G}_{hs}(e^{-it\mathbf{\hat{H}_{S}}}) -\mathscr{G}_{hs}(\mathbb{I}) =  \mathscr{G}_{hs}(e^{-it\mathbf{\hat{H}_{S}}})  $; i.e. the channel-complexity and the geometric complexity coincide in the limit where there is no environment! Before proceeding, we will present the following two propositions.
\begin{definition9}
\label{eqn:prop}
Let $\mathbf{\hat{H}}_{tot}=\mathbf{\hat{H}}_{S}+\mathbf{\hat{H}}_{I}+\mathbf{\hat{H}}_{E}$ Hermitian matrix acting on a tensor product Hilbert space $\mathscr{H}_{S}\otimes\mathscr{H}_{E}$, then following holds.
\begin{equation}
\mathscr{G}_{hs}(e^{-it\mathbf{\hat{H}}_{tot}}) \leq 
\mathscr{G}_{hs}(e^{-it\mathbf{\hat{H}}_{S}}) + \mathscr{G}_{hs}\big(e^{-it(\sqrt{\mathbf{\hat{H}}^{2}_{tot}-\mathbf{\hat{H}}_{S}^{2}})})  \end{equation}
\end{definition9}
\begin{proof}
\begin{equation}
\mathscr{G}_{hs}(e^{-it\mathbf{\hat{H}}_{tot}}) = \frac{t}{\sqrt{d^{2}-1}}\big\|\mathbf{\hat{H}}_{tot}\big\|_{H.S.} =
\end{equation}
\begin{equation}
\frac{t}{\sqrt{d^{2}-1}}\sqrt{|Tr\big\{\mathbf{\hat{H}}_{tot}^{2}-\mathbf{\hat{H}}_{S}^{2}+\mathbf{\hat{H}}_{S}^{2}\big\}|} \leq
\end{equation}
\begin{equation}
\frac{t}{\sqrt{d^{2}-1}}\sqrt{|Tr\big\{\mathbf{\hat{H}}_{tot}^{2}-\mathbf{\hat{H}}_{S}\big\}|+|Tr\big\{\mathbf{\hat{H}}_{S}^{2}\big\}|} \leq
\end{equation}
\begin{equation}
\frac{t}{\sqrt{d^{2}-1}}\sqrt{Tr\big\{|\mathbf{\hat{H}}_{tot}^{2}-\mathbf{\hat{H}}_{S}|\big\}|+Tr\big\{\mathbf{\hat{H}}_{S}^{2}\big\}} \leq
\end{equation}
\begin{equation}
\frac{t}{\sqrt{d^{2}-1}}\Bigg(\Big\|\sqrt{|\mathbf{\hat{H}}_{tot}^{2}-\mathbf{\hat{H}}_{S}^{2}|}\Big\|_{H.S.}+\Big\|\mathbf{\hat{H}}^{2}_{S}\Big\|_{H.S.}\Bigg) \leq
\end{equation}
\begin{equation}
\frac{t}{\sqrt{d^{2}-1}}\bigg(\Big\|\sqrt{|\mathbf{\hat{H}}_{tot}^{2}-\mathbf{\hat{H}}_{S}^{2}|}\big\|_{H.S.}+\big\|\mathbf{\hat{H}}_{S}\big\|_{H.S.}\Big) =
\end{equation}
\begin{equation}
\frac{1}{\sqrt{d^{2}-1}}\int_{0}^{t}\bigg(\Big\|\sqrt{|\mathbf{\hat{H}}_{tot}^{2}-\mathbf{\hat{H}}_{S}^{2}|}\Big\|_{H.S.}+\big\|\mathbf{\hat{H}}^{2}_{S}\big\|_{H.S.}\bigg)dt^{\prime} =
\end{equation} 
\begin{equation}
 \mathscr{G}_{hs}(e^{-it\sqrt{|\mathbf{\hat{H}}_{tot}^{2}-\mathbf{\hat{H}}_{S}^{2}|}}) + \mathscr{G}_{hs}(e^{-it\mathbf{\hat{H}}_{S}})  
\end{equation}
\end{proof}
With Remarks \ref{eqn:remark1} and \ref{eqn:remark2} along with Proposition \ref{eqn:prop} we easily obtain the following result. 
\begin{definition9}
\begin{equation}
\label{eqn:cc}
\mathcal{G}_{hs}(\Lambda_{t}) \leq \mathscr{G}_{hs}(e^{-it\mathbf{\hat{H}}_{S}})
\end{equation}
\end{definition9}
\begin{proof}
\begin{equation}
\mathcal{G}_{hs}\big(\Lambda_{t}\big):=\mathscr{G}_{hs}(e^{-it\boldsymbol{\hat{H}}_{tot}}) - \mathscr{G}_{hs}(e^{-it\big(\sqrt{|\mathbf{\hat{H}}_{tot}^{2}-\mathbf{\hat{H}}_{S}^{2}|}\big)})\leq
\end{equation}  
\begin{equation}
\mathscr{G}_{hs}(e^{-it\boldsymbol{\hat{H}}_{S}}) +\mathscr{G}_{hs}(e^{-it\sqrt{|\mathbf{\hat{H}}_{tot}^{2}-\mathbf{\hat{H}}_{S}^{2}|}})- 
\end{equation}
\begin{equation}
\mathscr{G}_{hs}(e^{-it\sqrt{|\mathbf{\hat{H}}_{tot}^{2}-\mathbf{\hat{H}}_{S}^{2}|}}) \leq 
\end{equation}
\begin{equation}
\mathscr{G}_{hs}(e^{-it\boldsymbol{\hat{H}}_{S}})
\end{equation}
\end{proof}
Indeed, this shows that there is loss of geometric complexity in the dynamics associated with the system when there are non-unitary effects present. 
Now, to further test the sensibility of the channel-complexity as a definition for geometric complexity, let us consider the case where  $\mathbf{\hat{H}}_{tot}= \mathbf{\hat{H}}_{S}+\mathbf{\hat{H}}_{E}$. Here, 
\begin{equation}
\mathcal{G}_{hs}(\Lambda_{t}) = \mathscr{G}_{hs}(e^{-it(\mathbf{\hat{H}}_{S}+\mathbf{\hat{H}}_{E})}) -   \mathscr{G}_{hs}(e^{-it\mathbf{\hat{H}}_{E}}) \leq
\end{equation}
\begin{equation}
 \mathscr{G}_{hs}(e^{-it\mathbf{\hat{H}}_{S}})
\end{equation}
where we have employed Remarks \ref{eqn:remark1} and \ref{eqn:remark2}. The latter confirms that in the case where there are no interactions between the system $E$ and the environment, the Channel-Complexity is bounded above via the geometric complexity of the dynamics generated by $\mathbf{\hat{H}}_{S}$ pertaining to the system.

Next, we verify that the Channel-Complexity contains geometrical information pertaining to the interaction between the system and the environment. To do this it suffices to consider the case where $\mathbf{\hat{H}}_{I} \neq 0$ and $\mathbf{\hat{H}}_{E} = 0 $. Here, $\sqrt{|\mathbf{\hat{H}}_{tot}^{2}-\mathbf{\hat{H}}_{S}^{2}|} = \sqrt{|\mathbf{\hat{H}}_{I}^{2}-\mathbf{\hat{H}}_{S}^{2}|} \neq \mathbf{\hat{H}}_{S}+\mathbf{\hat{H}}_{I}$. Finally, note that there is no geometric complexity when the only dynamics is given by the interaction Hamiltonian $\mathbf{\hat{H}}_{I}$. Let $\mathbf{\hat{H}}_{S} = 0 = \mathbf{\hat{H}}_{E}$, then the following holds.
\begin{equation}
\mathcal{G}_{hs}\big(\Lambda_{t}\big):=
\end{equation}
\begin{equation}
\mathscr{G}_{hs}(e^{-it\boldsymbol{\hat{H}}_{tot}}) - 
\mathscr{G}_{hs}(e^{-it\sqrt{|\mathbf{\hat{H}}_{tot}^{2}-\mathbf{\hat{H}}_{S}^{2}|}}) =    
\end{equation}
\begin{equation}
\mathscr{G}_{hs}(e^{-it\boldsymbol{\hat{H}}_{I}}\big) - 
\mathscr{G}_{hs}(e^{-it\sqrt{\mathbf{\hat{H}}_{I}^{2}}}) = 
\end{equation}
\begin{equation}
\mathscr{G}_{hs}(e^{-it\boldsymbol{\hat{H}}_{I}}) - \mathscr{G}_{hs}(e^{-it\mathbf{\hat{H}}_{I}}) = 0
\end{equation}

\section{Noise complexity} 
\;\;\; Quantum channels are used to model noise and dissipation. Considering the case where noise/dissipation can be modeled by an auxiliary system as we have done in the definition of $\Lambda_{t}(\cdot)$ in (\ref{eqn:nonuni2}), we may employ channel complexity to define an error/noise measure. 
\begin{definition}[Noise complexity]
\label{eqn:nscomp}
Let $\Lambda_{t}(\cdot)$ be defined as in (\ref{eqn:nonuni2}). We define the associated Noise complexity as follows.
\begin{equation}
\mathscr{N}_{hs}\big(\Lambda_{t}\big):= \big|\mathcal{G}_{hs}\big(\Lambda_{t}\big)-\mathscr{G}_{hs}(e^{-it\mathbf{\hat{H}}_{S}})\big|
\end{equation} 
\end{definition}
This is a measure of the amount of disturbance that is effectuated in the system $S$ from the quantum channel $\Lambda_{t}(\cdot)$. The advantage of such a measure is that it allows us to attribute a geometrical interpretation to the noise in terms of the notions of geometric complexity and channel-complexity; namely, the noise complexity $\mathscr{N}_{hs}(\Lambda_{t})$ may be interpreted as the channel-complexity of the "quantum" gate necessary to correct errors induced by the quantum channel $\Lambda_{t}(\cdot)$. Below, we present lower and upper bounds for the Noise Complexity.

\begin{definition9}
\begin{equation}
\big|\mathcal{G}_{hs}\big(\Lambda_{t}\big)-\mathscr{G}_{hs}(e^{-it\mathbf{\hat{H}}_{S}})\big| \geq
\end{equation}
\begin{equation}
\mathscr{G}_{hs}(e^{-it\big(\sqrt{|\mathbf{\hat{H}}_{tot}^{2}-\mathbf{\hat{H}}_{S}^{2}|}+|\mathbf{\hat{H}}_{S}|\big)})-\mathscr{G}_{hs}(e^{-it\boldsymbol{\hat{H}}_{tot}}\big)
\end{equation}
\end{definition9}
\begin{proof}
First note that,
\begin{equation}
\mathscr{G}_{hs}(e^{-it(\sqrt{|\mathbf{\hat{H}}_{tot}^{2}-\mathbf{\hat{H}}_{S}^{2}|}+|\mathbf{\hat{H}}_{S}|)})= 
\end{equation}
\begin{equation}
\frac{t}{\sqrt{d^{2}-1}}\Big\| \sqrt{|\mathbf{\hat{H}}_{tot}^{2}-\mathbf{\hat{H}}_{S}^{2}|}+|\mathbf{\hat{H}}_{S}|\Big\|_{H.S.}= 
\end{equation}
\begin{equation}
\frac{t}{\sqrt{d^{2}-1}}Tr\Big\{\mathbf{\hat{H}}_{tot}^{2}+\Big\{\sqrt{|\mathbf{\hat{H}}_{tot}^{2}-\mathbf{\hat{H}}_{S}^{2}|},\;|\mathbf{\hat{H}}_{S}|\Big\}  \Big\}  \geq 
\end{equation}
\begin{equation}
 \frac{t}{\sqrt{d^{2}-1}}Tr\big\{\mathbf{\hat{H}}_{tot}^{2}\big\} = \mathscr{G}_{hs}(e^{-it\mathbf{\hat{H}}_{tot}})  
\end{equation}
where we have used the fact that the anticommutaor of two positive operators is a positive operator in the last step. 

Now, using Remarks \ref{eqn:remark1} and \ref{eqn:remark2} we know that 
\begin{equation}
\mathscr{G}_{hs}(e^{-it\mathbf{\hat{H}}_{S}})+\mathscr{G}_{hs}(e^{\sqrt{|\mathbf{\hat{H}}_{tot}^{2}-\mathbf{\hat{H}}_{S}^{2}|}}) =
\end{equation}
\begin{equation}
\mathscr{G}_{hs}(e^{-it|\mathbf{\hat{H}}_{S}|})+\mathscr{G}_{hs}(e^{\sqrt{|\mathbf{\hat{H}}_{tot}^{2}-\mathbf{\hat{H}}_{S}^{2}|}}) \geq  
\end{equation}
\begin{equation}
\mathscr{G}_{hs}(e^{-it(\sqrt{|\mathbf{\hat{H}}_{tot}^{2}-\mathbf{\hat{H}}_{S}^{2}|}+|\mathbf{\hat{H}}|)})     
\end{equation}
From the above, it follows that
\begin{equation}
\big|\mathcal{G}_{hs}\big(\Lambda_{t}\big)-\mathscr{G}_{hs}(e^{-it\mathbf{\hat{H}}_{S}})\big| =  
\end{equation}
\begin{equation}
\mathscr{G}_{hs}(e^{-it\big(\sqrt{|\mathbf{\hat{H}}_{tot}^{2}-\mathbf{\hat{H}}_{S}^{2}|}\big)})+\mathscr{G}_{hs}(e^{-it\mathbf{\hat{H}}_{S}})-\mathscr{G}_{hs}(e^{-it\boldsymbol{\hat{H}}_{tot}}\big) =
\end{equation}
\begin{equation}
\mathscr{G}_{hs}(e^{-it\big(\sqrt{|\mathbf{\hat{H}}_{tot}^{2}-\mathbf{\hat{H}}_{S}^{2}|}\big)})+\mathscr{G}_{hs}(e^{-it|\mathbf{\hat{H}}_{S}|})-\mathscr{G}_{hs}(e^{-it\boldsymbol{\hat{H}}_{tot}}) \geq
\end{equation}
\begin{equation}
 \mathscr{G}_{hs}(e^{-it\big(\sqrt{|\mathbf{\hat{H}}_{tot}^{2}-\mathbf{\hat{H}}_{S}^{2}|}+|\mathbf{\hat{H}}_{S}|\big)})-\mathscr{G}_{hs}(e^{-it\boldsymbol{\hat{H}}_{tot}}) 
\end{equation}
\end{proof}
We may also bound the Noise Complexity from above; for this, we will need the following lemma, which is just more general exposition of Proposition \ref{eqn:prop} with a different proof. 
\begin{definition2}
\begin{equation}
\big| \mathscr{G}_{hs}(e^{-it\mathbf{\hat{A}}})-\mathscr{G}_{hs}(e^{-it\mathbf{\hat{B}}})\big|\leq \mathscr{G}_{hs}(e^{-it\sqrt{|\mathbf{\hat{A}}^{2}-\mathbf{\hat{B}}^{2}|}})
\end{equation}
\end{definition2}
\begin{proof}
\begin{equation}
\big| \mathscr{G}_{hs}(e^{-it\mathbf{\hat{A}}})-\mathscr{G}_{hs}(e^{-it\mathbf{\hat{B}}})\big| = 
\end{equation}
\begin{equation}
\frac{t}{\sqrt{d^{2}-1}}\Big|\big\|\mathbf{\hat{A}}\big\|_{H.S.}-\big\|\mathbf{\hat{B}}\big\|_{H.S.}\Big| = 
\end{equation}
\begin{equation}
\label{eqn:sqrin}
 \frac{t}{\sqrt{d^{2}-1}}\Big|\sqrt{Tr\big\{\mathbf{\hat{A}}^{2}\big\}}-\sqrt{Tr\big\{\mathbf{\hat{B}}^{2}\big\}}\Big| \leq  
 \end{equation}
 \begin{equation}
 \frac{t}{\sqrt{d^{2}-1}}\sqrt{\Big|Tr\big\{\mathbf{\hat{A}}^{2}\big\}-Tr\big\{\mathbf{\hat{B}}^{2}\big\}\Big|} =
\end{equation}
\begin{equation}
\frac{t}{\sqrt{d^{2}-1}}\sqrt{\Big|Tr\big\{\mathbf{\hat{A}}^{2}-\mathbf{\hat{B}}^{2}\big\}\Big|} \leq \frac{t}{\sqrt{d^{2}-1}}\sqrt{Tr\big\{\big|\mathbf{\hat{A}}^{2}-\mathbf{\hat{B}}^{2}\big|\big\}} =
\end{equation}
\begin{equation}
\frac{t}{\sqrt{d^{2}-1}}\Big\|\sqrt{\big|\mathbf{\hat{A}}^{2}-\mathbf{\hat{B}}^{2}\big|}\Big\|_{H.S.} = \mathscr{G}_{hs}(e^{-it\sqrt{|\mathbf{\hat{A}}^{2}-\mathbf{\hat{B}}^{2}|}}) 
\end{equation}
where we have used the fact that $|\sqrt{a}-\sqrt{b}|\leq \sqrt{|a-b|}$ for any real $a$ and $b$ in (\ref{eqn:sqrin}). 
\end{proof}
We now prove an upper bound to the Noise Complexity. 

\begin{Co}
\begin{equation}
\mathscr{N}_{hs}(\Lambda_{t}):=\big|\mathcal{G}_{hs}\big(\Lambda_{t}\big)-\mathscr{G}_{hs}(e^{-it\mathbf{\hat{H}}_{S}})\big|\leq
\end{equation}
 \begin{equation}
\mathscr{G}_{hs}(e^{-it\mathbf{\hat{H}}_{S}})- \mathscr{D}_{hs}\big( e^{-it\boldsymbol{\hat{H}}_{tot}}, e^{-it\big(\sqrt{|\mathbf{\hat{H}}_{tot}^{2}-\mathbf{\hat{H}}_{S}^{2}|}\big)}\big) 
\end{equation}
\end{Co}
\begin{proof}
\begin{equation}
\big|\mathcal{G}_{hs}\big(\Lambda_{t}\big)-\mathscr{G}_{hs}(e^{-it\mathbf{\hat{H}}_{S}})\big| = 
\end{equation}
\smaller
\begin{equation}
\mathscr{G}_{hs}(e^{-it\big(\sqrt{|\mathbf{\hat{H}}_{tot}^{2}-\mathbf{\hat{H}}_{S}^{2}|}\big)})+\mathscr{G}_{hs}(e^{-it\mathbf{\hat{H}}_{S}})-\mathscr{G}_{hs}\big(e^{-it\boldsymbol{\hat{H}}_{tot}}\big) =
\end{equation}
\begin{equation}
\mathscr{G}_{hs}(e^{-it\mathbf{\hat{H}}_{S}})-|\mathscr{G}_{hs}\big(e^{-it\boldsymbol{\hat{H}}_{tot}}\big)-\mathscr{G}_{hs}\big(e^{-it\big(\sqrt{|\mathbf{\hat{H}}_{tot}^{2}-\mathbf{\hat{H}}_{S}^{2}|}\big)}\big)| \leq    
\end{equation}
\begin{equation}
\label{eqn:boundff}
\mathscr{G}_{hs}(e^{-it\mathbf{\hat{H}}_{S}})- \mathscr{D}_{hs}\big( e^{-it\boldsymbol{\hat{H}}_{tot}}, e^{-it\big(\sqrt{|\mathbf{\hat{H}}_{tot}^{2}-\mathbf{\hat{H}}_{S}^{2}|}\big)}\big)  \end{equation}
\end{proof}
\normalsize

\section{Geometric complexity and random noise}
\label{eqn:RODES}
Until now, all of the dynamics we have considered have been deterministic. Indeed, quantum channels provide a deterministic means by which to model noise in a quantum system via connecting said system to an environment. There are, nevertheless, other mechanisms by which to model noise for quantum systems. In particular, one may utilize the Random Schr\"{o}dinger equation \cite{cz}, which is a Random ordinary differential equation (RODE). 
\begin{equation}
\label{eqn:RODE}
\partial_{t}\mathbf{\hat{U}}_{S}(t,\xi) = -i\big(\mathbf{\hat{H}}_{S}(t)+\mathbf{\hat{H}}_{R}(t, \xi)\big)\mathbf{\hat{U}}_{S}(t, \xi), 
\end{equation}
 \begin{equation}
\mathbf{\hat{U}}_{S}(0,\xi) = \mathbb{I}
\end{equation}
Here, $\mathbf{\hat{H}}_{S}(t)$ corresponds to the deterministic dynamics while $\mathbf{\hat{H}}_{R}(t,\xi)$ is a Hermitian matrix-valued random process which may depend on $\mathbf{\hat{H}}_{S}(t)$. We will denote $\mathbf{\hat{U}}_{S}(t,\xi)$ as the solution to the RODE (\ref{eqn:RODE}), including the term $\xi$ to emphasize that it is a random process. When there is no noise, we will just write $\mathbf{\hat{U}}_{S}(t)$ as the solution to (\ref{eqn:RODE}) for the case where $\mathbf{\hat{H}}_{R}(t,\xi) = 0 $, which is just the usual Schr\"{o}dinger equation.

In \cite{cz}, the authors develop a measure for the noise effectuated on a system via the Hermitian-matrix valued process $\mathbf{\hat{H}}_{R}(t,\xi)$ which leverages the Carnot- Carath\'eodory metric which conincides with the matric that we have been using in he previous sections up to an over all scaling. Namely, they study the fluctuations of the noisy dynamics, letting $\mathbf{\hat{V}}(t):= \mathbf{E}[\mathbf{\hat{U}}_{S}(t,\xi)]$, the measure of interest is $\mathscr{D}_{hs}(\mathbf{\hat{U}}(t,\xi),\mathbf{\hat{V}})$. Note however that
\begin{equation}
\label{eqn:upb}
\mathscr{D}_{hs}(\mathbf{\hat{U}}(t,\xi),\mathbf{\hat{U}}_{S}(t))\leq
\end{equation}
 \begin{equation}
\mathscr{D}_{hs}(\mathbf{\hat{U}}(t,\xi),\mathbf{\hat{V}}(t))+\mathscr{D}_{hs}(\mathbf{\hat{V}}(t),\mathbf{\hat{U}}(t))
\end{equation}
which gives us a way to estimate the error/perturbation caused by the noise on the system in a geometrical sense. From (\ref{eqn:upb}) we see that as the number of trajectories becomes large, the error $\mathscr{D}_{hs}(\mathbf{\hat{U}}(t,\xi),\mathbf{\hat{U}}_{S}(t))$ approaches $\mathscr{D}_{hs}(\mathbf{\hat{V}}(t),\mathbf{\hat{U}}(t))$ in probability; this can be shown using the usual central-limit-theorem-type-arguments from classical probability.
Of course, when the noise is such that $\mathbf{E}[\mathbf{\hat{U}}_{S}(t,\xi)] = \mathbf{\hat{U}}_{S}(t)$, we have $\mathscr{D}_{hs}(\mathbf{\hat{U}}(t,\xi),\mathbf{\hat{U}}_{S}(t)) = \mathscr{D}_{hs}(\mathbf{\hat{U}}(t,\xi),\mathbf{\hat{V}}(t)) $. For such a case $\mathscr{D}_{hs}(\mathbf{\hat{U}}(t,\xi),\mathbf{\hat{U}}_{S}(t))$ will then converge to $0$ in probability. This latter case is of particular interest for quantum computing, where one would like the statistics of a large sample size of quantum circuit outputs to converge to those characterized by a desired unitary operator $\mathbf{\hat{U}}_{S}(t)$ which characterizes the quantum algorithm being executed (excluding the measurement part). 

In \cite{cz} the authors furthermore narrow down the focus of their paper to consider the case where the deterministic term $\mathbf{\hat{H}}_{S}(t) = \sum_{j}h_{j}(t)\mathbf{\hat{H}}_{j}$ where $\{\mathbf{\hat{H}}_{i}\}_{i}$ are just elements of the Lie algebra $\mathfrak{su}(N)$ associated with the Lie group $SU(N)$ acting on an $N$-dimensional Hilbert space $\mathscr{H}$ for an $N$- dimensional quantum system. Then, the authors consider the case where the random term $\mathbf{\hat{H}}_{R}(t,\xi)$ satisfies the following bound.
\begin{equation}
ess\;sup\|\mathbf{\hat{H}}_{R}(t,\xi)\|_{H.S.} = \sqrt{\sum_{j}l_{j}h_{j}^{2}(t)}
\end{equation}
Under such conditions, it is shown that 
\begin{equation}
\mathscr{D}_{hs}(\mathbf{\hat{U}}_{S}(t,\xi), \mathbf{\hat{V}})\leq \int_{0}^{t}\|\mathbf{\hat{H}}_{R}(t,\xi)\|_{H.S.}ds \leq
\end{equation}
\begin{equation}\sqrt{\sum_{j}l_{j}h_{j}^{2}(t)} =  \int_{0}^{t}\sqrt{\langle\mathbf{\hat{H}}_{S}(t), \mathbf{\hat{H}}_{S}(t)\rangle_{\Omega}}ds
\end{equation}
where the Riemannian metric $\langle\cdot, \cdot\rangle_{\Omega}$ on $SU(N)$ is of the type already discussed in (\ref{eqn:metric}). Thus, noise is related to the warping of $SU(N)$ via replacing its usual metric (Hilbert, Schmidt metric) with the non-homogeneous one, $\langle\cdot,\cdot\rangle_{\Omega}$, in the estimation of the fluctuations in a $\mathbf{\hat{H}}_{R}(t,\xi)$-independent manner.

Let us now contrast the above with our definition for noise complexity, i.e. Definition \ref{eqn:nscomp}. There, we concluded with the bound
\begin{equation}
\mathscr{N}_{hs}(\Lambda_{t})\leq
\end{equation}
 \begin{equation}
\mathscr{G}_{hs}(e^{-it\mathbf{\hat{H}}_{S}})- \mathscr{D}_{hs}\big( e^{-it\boldsymbol{\hat{H}}_{tot}}, e^{-it\big(\sqrt{|\mathbf{\hat{H}}_{tot}^{2}-\mathbf{\hat{H}}_{S}^{2}|}\big)}\big)    
\end{equation}
which goes to zero when the interaction and environmental terms $\mathbf{\hat{H}}_{I}$ and $\mathbf{\hat{H}}_{E}$ are equal to zero. Using a similar approach, we will now study the difference between the geometric complexities $\mathscr{G}_{hs}(\mathbf{\hat{U}}_{S}(t))$, and $\mathscr{G}_{hs}(\mathbf{\hat{U}}_{S}(t,\xi))$ corresponding to the cases without and with noise respectively. Now, notice that 
\begin{equation}
\label{eqn:RODENISE}
 |\mathscr{G}_{hs}(\mathbf{\hat{U}}_{S}(t))-\mathscr{G}_{hs}(\mathbf{\hat{U}}_{S}(t,\xi))|\leq
 \end{equation}
 \begin{equation}\mathscr{D}_{hs}(\mathbf{\hat{U}}_{S}(t,\xi),\mathbf{\hat{U}}_{S}(t)) 
\end{equation}
which is the measure of noise studied in \cite{cz}, which was analyzed above. Although the frameworks are different, the techniques are the same; the estimation of the effects of the noise via measures such as the noise complexity and a similar measure (\ref{eqn:RODENISE}) are the heart of the matter.



\begin{section}{Time-dependent Hamiltonians}
\;\;\;A more general type of dynamics may be had if we now assume that the total Hamiltonian in (\ref{eqn:nonuni2}) and (\ref{eqn:nonuni}) is time-dependent. In previous sections we have already discussed the fact that a geodesic between $\mathbb{I}$ and $e^{-i\mathbf{\hat{H}}_{tot}(t)}$ may be parametrized as $\gamma(t)=\mathscr{\hat{T}}e^{-i\int_{0}^{t}\mathbf{\hat{H}}(s)ds}$ where all of the $\mathbf{\hat{H}}(s)$ constitute the set of Hermitian matrices that minimize the functional $\int_{0}^{t}\sqrt{\langle \mathbf{\hat{H}}(s),\mathbf{\hat{H}}(s)\rangle_{\Omega}}ds$ such that $\gamma(t)=\mathscr{\hat{T}}e^{-i\int_{0}^{t}\mathbf{\hat{H}}(s)}ds$ is indeed a geodesic between $\mathbf{I}$ and $e^{-i\mathbf{\hat{H}}_{tot}(t)}$ in $SU(N)$ equipped with the non-homogeneous metric $\langle\cdot, \cdot\rangle_{\Omega}$. In such a case, the definition of Channel-complexity generalizes to the following.

\begin{definition}[Channel-Complexity for time-dependent Hamiltonians]
Let $\Lambda_{t}(\cdot)$ be defined as in (\ref{eqn:nonuni}), but now characterized by the unitary operator $e^{-i\boldsymbol{\hat{H}}_{tot}(t)}$ which is generated by a time dependent Hamiltonian $\mathbf{\hat{H}}_{tot}(t)$. We will assume the decomposition $\mathbf{\hat{H}}_{tot}(t) = \int_{0}^{t}( \mathbf{\hat{H}}_{S}(s)+\mathbf{\hat{H}}_{I}(s)+\mathbf{\hat{H}}_{E}(s))ds$. Below we will just write $\mathbf{\hat{H}}(s):=\mathbf{\hat{H}}_{S}(s)+\mathbf{\hat{H}}_{I}(s)+\mathbf{\hat{H}}_{E}(s)$ for simplicity. The term $\mathbf{\hat{H}}_{S}(t)\otimes\mathbb{I}_{E}$, non-linear in $t$, corresponds to a geodesic in $SU(N)$ equipped with a particular non-homogeneous metric $\langle \cdot, \cdot\rangle_{\Omega}$ as we defined in (\ref{eqn:metric}). Furthermore, define $\|\mathbf{\hat{A}}\|_{\Omega}:=\sqrt{\langle\mathbf{\hat{A}}, \mathbf{\hat{A}}\rangle_{\Omega}}$. Finally, we define the Channel-Capacity in this case as follows
\begin{equation}
\mathcal{G}_{\Omega}\big(\Lambda_{t}\big):=
\end{equation}
\begin{equation}
\label{eqn:defchann}
 \frac{1}{\sqrt{d^{2}-1}}\int_{0}^{t}\bigg|\|\mathbf{\hat{H}}(s)\|_{\Omega}-\sqrt{\big|\|\mathbf{\hat{H}}^{2}(s)\|_{\Omega}-\|\mathbf{\hat{H}}^{2}_{S}(s)\|_{\Omega}\big|}\bigg|ds
\end{equation}
where $d$ is the dimension of the total Hilbert space and the Riemannian metric $\langle \cdot, \cdot\rangle_{\Omega}$ is the Hermitian form defined as follows.

Expressing $\mathbf{\hat{A}}$ and $\mathbf{\hat{B}}\in \mathfrak{su}(N)$ as column vectors, with each entry corresponding to a basis element of $\mathfrak{su}(N)$, i.e. 
\begin{equation}
\mathbf{Vec}(\mathbf{\hat{A}}) = \begin{pmatrix} A_{1} \\ \vdots\\ A_{N^{2}-1} \end{pmatrix} 
\;\;\;\mathbf{Vec}(\mathbf{\hat{B}}) = \begin{pmatrix} B_{1} \\ \vdots\\ B_{N^{2}-1} \end{pmatrix} 
\end{equation}
and letting 

\begin{equation}
\Omega = \begin{pmatrix}l_{1} & & \\ & \ddots & \\ & & l_{N^{2}-1}\end{pmatrix}
\end{equation}
Then,
\begin{equation}
 \big\langle \mathbf{\hat{A}},\mathbf{\hat{B}}\big\rangle_{\Omega}:=\frac{1}{d^{2}-1}\mathbf{Vec}(\mathbf{\hat{A}})^{\dagger}\Omega\mathbf{Vec}(\mathbf{\hat{B}}) 
\end{equation}
The subscript $\Omega$ signifies that we only consider the subspace associated with the environmental degrees of freedom. 

\end{definition}
The reason this definition for Channel Complexity looks rather different from the one we studied in earlier sections for time-independent $\mathbf{\hat{H}}_{tot}$ is due to the Hermitian form that one must introduce for the case of time-dependent Hamiltonians. We are still doing morally the same thing, i.e., ignoring information pertaining to the trajectory $\mathbf{\hat{H}}_{tot}(s)$, taken along the optimal trajectory for zero noise, and computing the difference between the lengths of the trajectory involving the full Hamiltonian and the trajectory which ignores the system Hamiltonians at every time $s$. This allows us to measure the influence of the terms attributed to noise, i.e. $\mathbf{\hat{H}}_{I}(s)$ and $\mathbf{\hat{H}}_{E}(s)$.

We now verify that this definition makes sense as a measure of complexity in more mathematical terms. Just as before, we expect to retrieve the usual geometric complexity in the limit $\mathbf{\hat{H}}_{I}(s)=0=\mathbf{\hat{H}}_{E}(s)$ for all $s$. Indeed, this is the case since then, the second term from the left in the integrand becomes zero and we are left with just the integral corresponding to the length of the minimal geodesic connecting the identity operator to the operator $\mathscr{\hat{T}}e^{-i\int_{0}^{t}\mathbf{\hat{H}}_{S}(s)ds}$. It is also simple to see that for $\mathbf{H}_{I}(s)$ being the only non-zero for all $s$, we have $\mathcal{G}_{\Omega}(\Lambda_{t} )= 0$. Now, we show that this generalized channel complexity is bounded above by the usual geometric complexity for the noiseless case. All we need to use is the fact that $\sqrt{|a-b|}\geq|\sqrt{a}-\sqrt{b}|$. Then, 
\begin{equation}
 \mathcal{G}_{\Omega}(\Lambda_{t}):=
 \end{equation}
 \small
 \begin{equation}
  \frac{1}{\sqrt{d^{2}-1}}\int_{0}^{t}\bigg|\|\mathbf{\hat{H}}(s)\|_{\Omega}-\sqrt{\big|\|\mathbf{\hat{H}}^{2}(s)\|_{\Omega}-\|\mathbf{\hat{H}}^{2}_{S}(s)\|_{\Omega}\big|}\bigg|ds \leq
\end{equation}
\normalsize
\begin{equation}
 \frac{1}{\sqrt{d^{2}-1}}\int_{0}^{t}\|\mathbf{\hat{H}}(s)\|_{\Omega}ds = \mathscr{G}_{\Omega}(\mathscr{\hat{T}}e^{-i\int_{0}^{t}\mathbf{\hat{H}}_{S}(s)ds})
\end{equation}
Similar properties and results may be proven for the case of a non-homogeneous metric $\langle\cdot, \cdot\rangle_{\Omega}$ as were proven for the case of time-independent $\mathbf{\hat{H}}_{tot}$ by using similar techniques.  We may now define a generalized version of the Noise-Complexity

\begin{definition}[Noise complexity for $\langle\cdot, \cdot\rangle_{\Omega}$]
\label{eqn:nscomp2}
Let $\Lambda_{t}(\cdot)$ be defined as in (\ref{eqn:nonuni2}) but with a time-dependent Hamiltonian $\mathbf{\hat{H}}_{tot}(t)$. We define the associated Noise complexity as follows.
\begin{equation}
\mathscr{N}_\Omega\big(\Lambda_{t}\big):= \big|\mathcal{G}_{\Omega}\big(\Lambda_{t}\big)-\mathscr{G}_{\Omega}(\mathscr{\hat{T}}e^{-i\int_{0}^{t}\mathbf{\hat{H}}_{s}(s)ds})\big|
\end{equation} 
\end{definition}

\section{A simple example}
\label{eqn:ex}
Consider the following total Hamiltonian which characterizes a perturbation on the systems by the environment; we will assume that the effect of the system on the environment is negligible and so there will be no self-Hamiltonian for the environment. 
\begin{equation}
\mathbf{\hat{H}}_{tot}=\mathbf{\hat{H}}_{S}+\varepsilon\mathbf{\hat{A}}_{S}\otimes \sum_{i}E_{i}\big|E_{i}\big\rangle\big\langle E_{i}\big|
\end{equation}
For simplicity assume that $E_{i}\geq 0$ and that both $\mathbf{\hat{H}}_{S}$ and $\mathbf{\hat{A}}_{S}$ are positive semidefinite operators that commute. For $\varepsilon>0$ small, and for a given $\boldsymbol{\hat{\rho}}\in \mathcal{S}(\mathscr{H}_{S}\otimes \mathscr{H}_{E})$ we have 
\begin{equation}
    Tr_{E}\big\{e^{-it\mathbf{\hat{H}}_{tot}}\boldsymbol{\hat{\rho}}e^{it\mathbf{\hat{H}}_{tot}}\big\} \approx 
\end{equation}
\begin{equation}
    e^{-it\mathbf{\hat{H}}_{S}}\Big(\sum_{i}\alpha_{i}e^{-i\varepsilon tE_{i}\mathbf{\hat{A}}_{S}}\boldsymbol{\hat{\rho}}e^{itE_{i}\mathbf{\hat{A}}_{S}}\Big)e^{i\varepsilon t\mathbf{\hat{H}}_{S}}
\end{equation}
The map 
\begin{equation}
\mathscr{N}_{t}(\cdot):=e^{-it\mathbf{\hat{H}}_{S}}\big(\cdot\big)e^{it\mathbf{\hat{H}}_{S}}
\end{equation}
is of course unitary, whilst the map 
\begin{equation}
\Lambda_{t}(\cdot):=\sum_{i}\alpha_{i}e^{-i\varepsilon tE_{i}\mathbf{\hat{A}}_{S}}(\cdot)e^{i\varepsilon tE_{i}\mathbf{\hat{A}}_{S}}
\end{equation}
is a non-unitary dephasing channel. 
\begin{equation}
\mathcal{G}_{hs}\big(\mathscr{N}_{t}\circ \Lambda_{t}\big(\boldsymbol{\hat{\rho}}\big)\big) = 
\end{equation}
\begin{equation}
\label{eqn:exa}
 \frac{t}{\sqrt{d^{2}_{tot}-1}}\Big(\|\mathbf{\hat{H}}_{S}\|_{\Omega}(1-\sqrt{\varepsilon}\Omega\big(1-\sqrt{\varepsilon}\Omega\big)\Big) +\mathcal{O}(\varepsilon^{3/2})  
\end{equation}
Where $\Omega:=\frac{\sqrt{ (d^{2}_{tot}-1)\langle\mathbf{\hat{A}}_{S},\mathbf{\hat{H}}_{S}\rangle_{hs}\langle E_{i}\rangle}}{\|\mathbf{\hat{H}}_{S}\|_{\Omega}}$. For a proof of the latter please see Appendix \ref{eqn:apexa}. For $\varepsilon$ small enough that the terms of order $3/2$ and greater are negligible, this leads to a decrease in geometric complexity whenever $\varepsilon\leq \frac{\|\mathbf{\hat{H}}_{S}\|_{\Omega}^{2}}{(d^{2}_{tot}-1)\langle\mathbf{\hat{A}}_{S},\mathbf{\hat{H}}_{S}\rangle_{hs}\langle E_{i}\rangle}$. This is a reasonable assumption since Environments tend to be much larger than the system and therefore $\frac{\|\mathbf{\hat{H}}_{S}\|_{\Omega}^{2}}{ (d^{2}_{tot}-1)\langle\mathbf{\hat{A}}_{S},\mathbf{\hat{H}}_{S}\rangle_{hs}\langle E_{i}\rangle}<< 1$. Indeed, above we have used the fact that the anticommutator of two positive semidefinite operators that commute is another positive semidefinite operator as well as the fact that the tensor product of two positive semidefinite operators is yet another positive semidefinite operator. Notice that the term following the minus sign in equation (\ref{eqn:exa}) will always be positive because the Hilbert-Schmidt inner product of two positive semidefinite operators is never negative. In this case the noise-complexity is just $\mathcal{N}(\Lambda_{t})\approx \frac{t}{\sqrt{d_{tot}^{2}-1}}\sqrt{\varepsilon}\Omega(1-\sqrt{\varepsilon}\Omega)$

\section{An abstract definition of a noisy quantum computer}
\;\;\; Recently, efforts have been made to develop mathematical frameworks that abstract the notion of a quantum computer. Inciting such work raises questions regarding what a fundamental computation unit should be; in theory, although a quantum computation algorithm should be implemented via a unitary matrix representing a quantum circuit, it is nevertheless nebulous what the fundamental unit of computation should be. This leads to the disparity amongst members of the quantum computation community when it comes to questions regarding the complexity of a quantum algorithm. Since Nielsen's original work on quantum computational complexity \cite{NielGeo}, efforts have been made to drift away from the traditional conceptualization of a quantum algorithm as a quantum circuit with gates and depth, to one where the physical restrictions in the hardware are characterized by non-homogeneous Riemannian metrics on $SU(N)$; the non-homogeneity, as that exemplified by the metrics introduced earlier in this paper, see (\ref{eqn:metric}), $\langle \cdot, \cdot\rangle_{\Omega}$ may be tuned in order to make dynamics involving more than one and two qubit more costly for example for example. This geometric definition of quantum computational complexity of a quantum algorithm is seamlessly applicable within the framework of quantum annealing, a framework equivalent to that of quantum gates \cite{annealing}, this is not so much the case for the quantum gates framework, which would first require us to deduce the time-dependent Hamiltonian describing the quantum algorithm in question.

Using the framework presented in \cite{Daniela}, we will present here a more general mathematical definition for a quantum computer. This definition will contain two notions of complexity, the algebraic one proposed in \cite{Daniela}, and the geometric one which we have been using in this paper, attributed to \cite{NielGeo}. First, we will highlight the relevant definitions and results from \cite{Daniela}, and then we present our generalization.

\begin{definition} A matrix $\mathbf{\hat{A}}\in \mathbb{M}_{N}(\mathbb{C})$ is called a \emph{random variable} in the algebraic probability space $\big( \mathbb{M}_{N}(\mathbb{C}), \mathbf{\hat{\rho}} \big)$ if $\mathbf{\hat{A}}$ is Hermitian. The spectral theorem then characterizes the events $\big\{\mathbf{\hat{A}} = x \big\}$ as follows. 
\begin{equation}
\mathbf{\hat{P}}_{\big\{\mathbf{\hat{A}}=x\big\}}:= \begin{array}{cc}
  \Bigg\{ & 
    \begin{array}{cc}
\mathbf{\hat{P}}_{x},& if \;\; x\in\sigma(\mathbf{\hat{A}}) \\
\mathbf{\hat{0}}, &\;\; otherwise 
\end{array}
\end{array}
\end{equation}
\end{definition}
Where $\mathbf{\hat{P}}_{x}$ projects to the eigenspace associated with $x$ when $x\in\sigma(\mathbf{\hat{A}})$. The events $\big\{\mathbf{\hat{A}}<x\big\}$ and $\big\{\mathbf{\hat{A}}>x\big\}$ may be defined in a similar fashion. The law of the random variable $\mathbf{\hat{A}}$ in the algebraic probability space $\big( \mathbb{M}_{N}(\mathbb{C}), \mathbf{\hat{\rho}} \big)$ is defined as the function 
\begin{equation}
\mathbb{P}_{\mathbf{\hat{\rho}}}(\mathbf{\hat{A}}= x):= Tr\big\{\mathbf{\hat{\rho}}\mathbf{\hat{P}}_{x}\big\}
\end{equation}
The projectors $\mathbf{\hat{P}}_{x}$ will form a POVM for the case where the random variable $\mathbf{\hat{A}}$ has full support. This means that $\sum_{\lambda\in\sigma(\mathbf{\hat{A}})}\mathbb{P}_{\mathbf{\hat{\rho}}}(\mathbf{\hat{A}}=\lambda) = Tr\big\{\boldsymbol{\hat{\rho}}\sum_{\lambda\in\sigma(\mathbf{\hat{A}})}\mathbf{\hat{P}}_{\lambda}\big\} = Tr\big\{\boldsymbol{\hat{\rho}}\big\} = 1$.
\begin{definition}{Single-qubit quantum processor unit:}
A single-qubit quantum processor unit is defined as a random variable $\mathbf{\hat{A}}$ in the algebraic probability space $\big(\mathbb{M}_{2}(\mathbb{C}), \mathbf{\hat{\rho}}\big),$ where $\mathbf{\hat{\rho}}$ is a pure state. 
\end{definition} 
\begin{definition}{An n-qubit Quantum Processing Unit (QPU):}
An n-qubit QPU is defined by the tensor product of n-single qubit QPU $\bigotimes_{k=1}^{n}\mathbf{\hat{A}_{k}}$, where the $\mathbf{\hat{A}}_{k}$ are all single-qubit QPU. Given the random variables $\mathbf{\hat{A}}_{k}$ in the corresponding algebraic probability space $\big(\mathbb{M}_{2}(\mathbb{C}),\mathbf{\hat{\rho}}_{k}\big)$, where $\mathbf{\hat{\rho}}_{k} = |\psi_{k}\rangle \langle\psi_{k}| \in \mathcal{S}\big(\mathbb{M}_{2}(\mathbb{C})\big)$ for some unit vector $|\psi_{k}\rangle \in \mathbb{CP}$, then $\bigotimes_{k=1}^{n}\mathbf{\hat{A}}_{k}$ is a random variable in the algebraic probability space $\big(\mathbb{M}_{2^{n}}(\mathbb{C}), \bigotimes_{k=1}^{n}\mathbf{\hat{\rho}}_{k}\big)$, where $\mathbf{\hat{\rho}}_{k}\in\mathcal{S}_{1}\big(\mathbb{M}_{2^{n}}(\mathbb{C})\big). $
\end{definition}
\begin{definition}{Universal digital quantum computer (UDQC):}
A universal digital quantum computer (UDQC) is a triple $(\mathbf{\hat{A}},\boldsymbol{\hat{\rho}}_{0}, \mathscr{N}_{t})$ where:
\begin{itemize}
\item $\mathbf{\hat{A}}$ is an $n-$qubit QPU with 
\begin{equation}
\mathbf{\hat{A}} = \mathbf{\hat{U}}\Big(\sum_{k\in\mathbb{Z}_{N}}\lambda_{b_{n}(k)}|b_{n}(k)\rangle\langle b_{n}(k)|\Big)\mathbf{\hat{U}}^{\dagger}
\end{equation}  
for some unitary matrix $\mathbf{\hat{U}}\in SU(N)$.

\item $\mathbf{\hat{\rho}}_{0} = |b_{n}(0)\rangle\langle b_{n}(0)|$ is the fixed initial pure state, and 

\item 
$\mathscr{N}_{t}$ is a unitary quantum channel which solves the von Neumann Liouiville equation, i.e. 
\begin{equation}
i\partial_{t}\mathscr{N}_{t}(\boldsymbol{\hat{\rho}}_{0})= \big[\mathbf{\hat{H}}, \mathscr{N}_{t}(\boldsymbol{\hat{\rho}}_{0})\big]
\end{equation}
where $\mathbf{\hat{H}}$ is the generator of the unitary map $\mathscr{N}_{t}$. 
\end{itemize}
\end{definition}
We now extend the latter definition to include the noisy channels. 
\begin{definition}{Noisy Universal digital quantum computer (NUDQC):} This is just a UDCQ as defined above but now the channel $\mathscr{N}_{t}$ will be replaced by the channel $\Lambda_{t}$, as defined in (\ref{eqn:nonuni2}), with a time-dependent $\mathbf{\hat{H}}_{tot}$.
\end{definition}
In \cite{Daniela}, the notion of a Quantum Computational Procedure is also introduced. The starting point for this notion is a UDQC $(\mathbf{\hat{A}}, \boldsymbol{\hat{\rho}}_{0}, \mathscr{N}_{t})$, and a uitary matrix $\mathbf{\hat{U}}(t)$ generating the map $\mathscr{N}_{t}$. The corresponding quantum computational procedure is hence composed of two steps:
\begin{itemize}
\item Quantum Ciruit Process: A quantum state evolution applies $\mathbf{\hat{U}}(t)$ to the initial state $\boldsymbol{\hat{\rho}}_{0}$ via the map $\mathscr{N}_{t}(\boldsymbol{\hat{\rho}}_{0}):=\mathbf{\hat{U}}(t)\boldsymbol{\hat{\rho}}_{0}\mathbf{\hat{U}}^{\dagger}(t)$, yielding the random variable $\mathbf{\hat{A}}$ in the output algebraic probability space at time $t$ $ (\mathbb{M}_{N}(\mathbb{C}), \mathscr{N}_{t}(\boldsymbol{\hat{\rho}}_{0}))$.
\item  Measurement process: A measurement of $\mathbf{\hat{A}}$ in the output space $(\mathbb{M}_{N}(\mathbb{C}), \mathscr{N}_{t}(\boldsymbol{\hat{\rho}}_{0}))$ reveals the outcome $k \in \mathbb{Z}_{N}$ with probability $\mathbb{P}_{\mathscr{N}_{t}(\boldsymbol{\hat{\rho}}_{0})}(\mathbf{\hat{A}}=k)$.
\end{itemize}

For the case of NUDQC, we now add a third step to the Quantum Computational Procedure notion, consisting of an error mitigation process. The literature surrounding error correction is vast; there is no one-size-fits-all scheme that unites all techniques. A popular approach to quantum error correction involves invertible noisy quantum channels \cite{laflame}. Indeed, when there is invertibility, one can construct a corresponding channel that mitigates the effects of the noise. We will bypass all ambiguity by sticking to a purely geometric notion of quantum error correction, whose utility is to quantify the complexity of an error-correcting code/ mitigating process in geometric terms. The goal of all error correction codes/schemes is to take a state $\boldsymbol{\hat{\rho}}_{noise}$ which has undergone effects due to some noisy channel to the actual state $\boldsymbol{\hat{\rho}}_{actual}$ without the noise. To do this optimally we must move $\boldsymbol{\hat{\rho}}_{noise}$ to $\boldsymbol{\hat\rho}_{actual}$ along a geodesic in $\mathcal{S}(\mathscr{H}_{2^n})$ utilizing the appropriate Riemannian metric, in this case the so-called Bures metric. Using our purely geometric approach, we are able to quantify the error due to noisy channels with the bound (\ref{eqn:boundff}). When one may not construct a channel that takes us from $\boldsymbol{\hat{\rho}}_{noise}$ to $\boldsymbol{\hat\rho}_{actual}$ along a geodesic, one then resorts to statistical inference; namely, with multiple shots of the Noisy Quantum Computational Process in question we may deduce the underlying quantum channel $\mathscr{N}_{t}$ and retrieve $\boldsymbol{\hat{\rho}}_{actual}$ as the variance of the error decrease with the number of shots. 
\begin{definition}{Noisy Quantum Computational Procedure} 
A Noisy Quantum computational Procedure (NQCP) is a QCP with the unitary channel $\mathscr{N}_{t}(\cdot)$ traded in for a non-unitary channel $\Lambda_{t}(\cdot)$. This procedure includes an error mitigation step, which will depend on the scenario. The magnitude of the errors may be measured using the Noise-Complexity defined in Definition \ref{eqn:nscomp2} or the analogous measure for the case where $\Lambda_{t}(\cdot)$ is obtain from a time-dependent Hamiltonian $\mathbf{\hat{H}}_{tot}(t)$.
\end{definition}
\;\;\;In \cite{Daniela}, the concept of an Elementary Quantum Gates is also discussed. This is done to create a framework that quantifies the complexity of a quantum algorithm in algebraic terms. Below, we present the definition for an elementary quantum gate and an important result from \cite{Daniela} that bounds the algebraic complexity of a quantum algorithm described by a unitary operator $\mathbf{\hat{U}} \in SU(N)$ for an appropriate $N$. 

\begin{definition}{Elementary Quantum Gate:} A unitary matrix $\mathbf{\hat{U}}\in SU(N)$ is called an elementary quantum gate if there exists a pair $(\mathbf{\hat{V}}, i)\in SU(2)\times Emb^{*}(SU(2), SU(N))$ such that $\mathbf{\hat{U}}=i(\mathbf{\hat{V}})$. The set of al elementary quantum gates in $SU(N)$ is denoted by $QG(N)$. Where 
\begin{equation}
Emb^{*}(SU(2),SU(N)) :=
\end{equation}
\smaller
\begin{equation}
\big\{f \in Emb(SU(2), SU(N)) \big| f(\mathbf{\hat{U}})^{\dagger} = f(\mathbf{\hat{U}}^{\dagger}) \;\;\forall\;\; \mathbf{\hat{U}}\in SU(2)\big\}
\end{equation}
\normalsize
\end{definition}
The main result of \cite{Daniela} is the following theorem which has twofold ramifications; firstly, it shows that the set $QG(N)$ is universal and secondly, it gives an upper bound to the algebraic-complexity of a quantum algorithm described by a unitary matrix $\mathbf{\hat{U}}\in SU(N)$.
\begin{definition4}{Universality of $QG(N)$ and upper bound on algebraic complexity: }
The set of quantum gates $QG(N)$ forms a universal dictionary for the unitary group $SU(N)$; that is, for every $U\in SU(N)$, there exists an integer $m=\frac{N(N-1)}{2}$ such that $\mathbf{\hat{U}}$ can be expressed as a product of $m$ elementary quantum gates. 
\end{definition4}
We are now ready to formally define algebraic complexity.
\begin{definition}{Algebraic-Complexity:}
A quantum circuit of length $l\geq 0$ (algebraic-complexity $l$) for a unitary matrix $\mathbf{\hat{U}}\in SU(N)$ is defined as follows: $l = 0$ iff $\mathbf{\hat{U}} = \mathbb{I}_{N}$. Otherwise, a quantum circuit is a finite sequence of elementary quantum gates $\big\{\mathbf{\hat{U}}_{1}, .., \mathbf{\hat{U}}_{l}\big\} \subset QG(N)\big\{\mathbb{I}_{N}\big\}$, satisfying 
\begin{equation}
\label{eqn:prod}
\mathbf{\hat{U}} = \mathbf{\hat{U}}_{l}\mathbf{\hat{U}}_{l-1}\cdot\cdots\mathbf{\hat{U}}_{1}
\end{equation}
with the additional condition that $\mathbf{\hat{U}}_{i+1}\mathbf{\hat{U}}_{i}\neq\mathbb{I}_{N}$ for all $1\leq i< l$.
This implies that if $\mathbf{\hat{U}}\in QG(N)$ has circuit length $l\geq 1$, then for each $1\leq k \leq l$, there exists a pair $(\mathbf{\hat{V}}_{k}, i_{k})\in SU(2)\times Emb^{*}(SU(2), SU(N))$ such that $\mathbf{\hat{U}}_{k}=i_{k}(\mathbf{\hat{V_{k}}})$, and 
\begin{equation}
\mathbf{\hat{U}} = i_{l}(\mathbf{\hat{V}}_{l})i_{l-1}(\mathbf{\hat{V}})\cdots i_{1}(\mathbf{\hat{V}}_{1})   
\end{equation}
\end{definition}

The algebraic complexity gives a simple scheme by which to compare the costliness of various quantum algorithms. For the cases where the quantum gate formalism of quantum computing is employed, this translates seamlessly, and one needs only derive a decomposition of the form in (\ref{eqn:prod}) in order to obtain upper bounds on the algebraic-complexity. The fundamental resource is defined as unitary operators acting on a single qubit, which are just operations described by symmetries on the Bloch-Sphere. On the other hand, it is not clear how one might incorporate noise in such a definition of computational complexity in such a way that lower complexities ensue from noisy channels; in previous sections this was one of the desired properties for a generalization of the concept of geometric complexity for the cases where non-unitary maps were considered. It is clear that for noisy channels, using the RODE formalism presented in Section \ref{eqn:RODES}, multiple runs of the algorithm described by (\ref{eqn:prod}) would have to be executed to reduce statistical error, making the cost a multiple of $l$; furthermore, determining the complexity reduction in the sense characterized by the Channel Complexity in general with time-dependent $\mathbf{\hat{H}}_{tot}(t)$ is elusive. 

Just as the notion of the Elementary Quantum Gate constitutes the key ingredient in defining algebraic-complexity, the Carnot-Carath\'eodory metric on $SU(N)$ constructed from the Riemannian metric $\langle\cdot, \cdot\rangle_{\Omega}$ defined in (\ref{eqn:metric}) is fundamental in defining geometric complexity. To contrast the notion of Elementary Quantum Gate, we now coin the notion of Elementary-Constraints-Metric.
\begin{definition}{Elementary-Constraints-Metric for a quantum computation procedure: }
Let $\mathbf{\hat{U}}$ be any quantum computation procedure. Then, the Elementary-Constraints-Metric (ECM) is the  Carnot-Carath\'eodory metric $\langle\cdot,\cdot\rangle_{\Omega}$ characterized by the non-homogeneous Riemannian metric $\langle \cdot, \cdot\rangle_{\Omega}$ whose non-homogeneity is due to physical constraints; e.g. often only gates involving one or two spins are implementable. These metrics are assumed to be diagonal with respect to the spin-string basis; indeed, the spin-string basis is the basis upon which the elementary constraints shall be imposed; hence the "Elementary" in the name, i.e. spin-spin interactions are the fundamental interactions in quantum computing. 
\end{definition}
\end{section}
\section{Conclusion}
\;\;\; In this work, we have developed a geometric framework for analyzing the complexity of quantum channels, extending the previously established notions of geometric complexity in works such \cite{NielGeo}\cite{dorth1}\cite{dorth2}\cite{brandt}\cite{Jeff}, a notion pioneered by Michael Nielsen  to encompass general quantum processes. By invoking the theories Lie groups, and Riemannian geometry, we introduced a definition of \emph{channel complexity} that captures the  geometric cost associated with implementing a given quantum channel; in particular, such a measure may be use to calculate the reduction of geometric complexity induced by noisy channels as was exhibited in Section \ref{eqn:ex}. This formulation incorporates both coherent and dissipative non-unitary dynamics.

The primary contribution of this work has been to provide a measure of complexity that is geometrical and is bounded above by the geometric complexity of Nielsen associated with the self-dynamics of the system with no noise. Other related concepts, such as the noise complexity were also presented. Furthermore, the channel complexity was motivated via leveraging the concept of coherence capacity, presented in Section \ref{eqn:sec2}. In the process of presenting the notion of coherence capacity we have also proved Theorem \ref{eqn:devoheringpowerbound} which is a tighter bound to the analogous Theorem proven in \cite{Bu} via different techniques. Finally, we used the notion of Geometric Complexity to broaden some of the definitions of Quantum Computational Procedures presented in \cite{Daniela} to include geometric complexity on top of the algebraic notion of complexity already present.

Looking toward the future, it would be interesting to analyze how our geometric formalism applies to \emph{fault-tolerant quantum computation} and \emph{quantum error correction}. Indeed, error correction/ mitigation procedures are costly and it would be of interest to quantify said cost via geometric means. Furthermore, the interplay between \emph{channel complexity} and \emph{quantum scrambling} or \emph{quantum chaos} remains largely unexplored in the context of open quantum systems. Extending notions of operator growth and complexity dynamics \cite{cot} \cite{B2} to non-unitary evolutions may offer new diagnostics for ergodicity, thermalization, and quantum chaos in many-body systems; this will inevitably heavily involve the Hermitian forms that we have been using as our Riemannian metrics $\langle \cdot, \cdot\rangle_{\Omega}$ in a very active manner.  

While our definitions are formally well-founded, algorithms for estimating channel complexity have yet to be explored. Future work consecrated to the analysis of specific models via numerical methods would be worth the effort. Moreover, extending our geometric approach to \emph{continuous-variable systems} and \emph{quantum field theories} would allow for broader applicability in quantum optics and relativistic settings. This may involve generalizing the purification framework to infinite-dimensional Hilbert spaces and understanding how complexity behaves under spacetime symmetries.However, such generalizations might require us to leave the realm of Riemannian geometry and enter the field of Finsler manifolds and or Banach manifolds since the underlying Hilbert spaces would now be infinite-dimensional. Finally, future work could also explore embedding our formalism within a broader \emph{resource-theoretic} or \emph{categorical} framework \cite{gour}, further exploring the relationship between geometric complexity and other resource measures like entanglement, coherence, and magic.

\end{multicols}

\newpage

\appendix 
\addcontentsline{toc}{section}{Appendices}
\section{Proof of Theorem \ref{eqn:ratethe}}
\label{eqn:ratethemdem}
\begin{proof}
\begin{equation}
   R_{\mathscr{E}}(\mathbf{\hat{H}}(t),\boldsymbol{\hat{\rho}}) = \frac{d}{dt}C_{\mathscr{E}}\big(\boldsymbol{\hat{\rho}}_{t}\big) = \frac{d}{dt}\bigg(\zeta\big(\boldsymbol{\hat{\rho}}_{t}\big) - \zeta\big(\mathscr{E}\big(\boldsymbol{\hat{\rho}}_{t}\big)\big)\bigg) = \frac{d}{dt}Tr\big\{\boldsymbol{\hat{\rho}}_{t}^{2}\big\} -\frac{d}{dt}Tr\big\{\mathscr{E}\big(\boldsymbol{\hat{\rho}}_{t}\big)^{2}\big\} = 
\end{equation}
\begin{equation}
2Tr\big\{\boldsymbol{\hat{\rho}}_{t}\frac{d}{dt}\boldsymbol{\hat{\rho}}_{t}\big\} -2Tr\big\{\mathscr{E}\big(\boldsymbol{\hat{\rho}}_{t}\big)\frac{d}{dt}\mathscr{E}\big(\boldsymbol{\hat{\rho}}_{t}\big)\big\} =  2\bigg(Tr\big\{\boldsymbol{\hat{\rho}}_{t}\frac{d}{dt}\boldsymbol{\hat{\rho}}_{t}\big\} -Tr\big\{\mathscr{E}\big(\boldsymbol{\hat{\rho}}_{t}\big)\mathscr{E}\big(\frac{d}{dt}\boldsymbol{\hat{\rho}}_{t}\big)\big\}\bigg) =   
\end{equation}
\begin{equation}
-2i\bigg(Tr\big\{\boldsymbol{\hat{\rho}}_{t}\big[\mathbf{\hat{H}}(t),\boldsymbol{\hat{\rho}}_{t}\big]\big\} -Tr\big\{\mathscr{E}\big(\boldsymbol{\hat{\rho}}_{t}\big)\mathscr{E}\big(\big[\mathbf{\hat{H}}(t),\boldsymbol{\hat{\rho}}_{t}\big]\big)\big\}\bigg) =
 \end{equation}
\begin{equation}
2i\bigg(Tr\big\{\mathscr{E}\big(\boldsymbol{\hat{\rho}}_{t}\big)\mathscr{E}\big(\big[\mathbf{\hat{H}}(t),\boldsymbol{\hat{\rho}}_{t}\big]\big)\big\}\bigg)
\end{equation}
since $Tr\big\{\mathbf{\hat{A}}\big[\mathbf{\hat{B}},\mathbf{\hat{A}}\big]\big\} = Tr\big\{\mathbf{\hat{A}}\mathbf{\hat{B}}\mathbf{\hat{A}}- \mathbf{\hat{A}}\mathbf{\hat{A}}\mathbf{\hat{B}}\big\} = Tr\big\{\mathbf{\hat{A}}\mathbf{\hat{B}}\mathbf{\hat{A}}-\mathbf{\hat{A}}\mathbf{\hat{B}}\mathbf{\hat{A}}\big\} = 0$ by employment of the cycliciy of the trace in the final step. Furthermore, note that
\begin{equation}
Tr\big\{\mathscr{E}\big(\boldsymbol{\hat{\rho}}_{t}\big)\mathscr{E}\big(\big[\mathbf{\hat{H}}(t),\boldsymbol{\hat{\rho}}_{t}\big]\big)\big\} = Tr\big\{ \sum_{i}\mathbf{\hat{P}}_{i}\boldsymbol{\hat{\rho}}_{t}\mathbf{\hat{P}}_{i}\big[\mathbf{\hat{H}}(t),\boldsymbol{\hat{\rho}}_{t}\big]\mathbf{\hat{P}}_{i}\big\}=
\end{equation}
\begin{equation}
\sum_{i}Tr\big\{\mathbf{\hat{P}}_{i} \mathbf{\hat{P}}_{i}\boldsymbol{\hat{\rho}}_{t}\mathbf{\hat{P}}_{i}\big[\mathbf{\hat{H}}(t),\boldsymbol{\hat{\rho}}_{t}\big]\big\} = Tr\big\{\mathscr{E}\big(\boldsymbol{\hat{\rho}}_{t}\big)\big[\mathbf{\hat{H}}(t),\boldsymbol{\hat{\rho}}_{t}\big]\big\} = 
\end{equation}
\begin{equation}
 Tr\big\{\mathscr{E}\big(\boldsymbol{\hat{\rho}}_{t}\big)\mathbf{\hat{H}}(t)\boldsymbol{\hat{\rho}}_{t}-\mathscr{E}\big(\boldsymbol{\hat{\rho}}_{t}\big)\boldsymbol{\hat{\rho}}_{t}\mathbf{\hat{H}}(t)\big\}  = Tr\big\{\boldsymbol{\hat{\rho}}_{t}\mathscr{E}\big(\boldsymbol{\hat{\rho}}_{t}\big)\mathbf{\hat{H}}(t)-\mathscr{E}\big(\boldsymbol{\hat{\rho}}_{t}\big)\boldsymbol{\hat{\rho}}_{t}\mathbf{\hat{H}}(t)\big\} =  
\end{equation}
\begin{equation}
Tr\big\{\big[\boldsymbol{\hat{\rho}}_{t},\mathscr{E}\big(\boldsymbol{\hat{\rho}}_{t}\big)\big]\mathbf{\hat{H}}(t)\big\}
\end{equation}
Hence, 
\begin{equation}
R_{\mathscr{E}}(\mathbf{H}(s),\boldsymbol{\hat{\rho}})\Big|_{t=t^\prime}=2iTr\big\{\big[\boldsymbol{\hat{\rho}}_{t^\prime},\mathscr{E}\big(\boldsymbol{\hat{\rho}}_{t^{\prime}}\big)\big]\mathbf{H}(t^\prime)\big\}
\end{equation}
which proves the first part of the proposition. Now, a simple application of H\"{o}lder's inequality for Schatten classes (i.e. trace class ideals), namely $\|\mathbf{\hat{A}}\mathbf{\hat{B}}\|_{1}\leq\|\mathbf{\hat{A}}\|_{p}\|\mathbf{\hat{B}}\|_{q} ,\frac{1}{p}+\frac{1}{q} = 1$, along with the fact that $|Tr\{\mathbf{\hat{A}}\}|\leq \|\mathbf{\hat{A}}\|_{1}$ yields the bound 
\begin{equation}
\big|R_{\mathscr{E}}(\mathbf{\hat{H}}(s),\boldsymbol{\hat{\rho}})\big|\Big|_{t=t^\prime}\leq \big\|\big[\boldsymbol{\hat{\rho}}_{t^\prime},\mathscr{E}\big(\boldsymbol{\hat{\rho}}_{t^{\prime}}\big)\big]\big\|_{H.S.}\big\|\mathbf{\hat{H}}(t^{\prime})\big\|_{H.S.}
\end{equation}
where $\|\cdot\|_{H.S}:=\|\cdot\|_{2}$ ( we use the subscript $H.S.$ to represent Hilbert-Schmidt). It is obvious that $\big\|\big[\boldsymbol{\hat{\rho}}_{t^\prime},\mathscr{E}\big(\boldsymbol{\hat{\rho}}_{t^{\prime}}\big)\big]\big\|_{H.S.}\leq\sup_{t^\prime\in[0,t]}\big\|\big[\boldsymbol{\hat{\rho}}_{t^\prime},\mathscr{E}\big(\boldsymbol{\hat{\rho}}_{t^{\prime}}\big)\big]\big\|_{H.S.}$, we need therefore only show that $\big\|\big[\boldsymbol{\hat{\rho}}_{t^\prime},\mathscr{E}\big(\boldsymbol{\hat{\rho}}_{t^{\prime}}\big)\big]\big\|_{H.S.}\leq \sqrt{2}$ for all $t^\prime\in[0,t]$ in order to complete the proof.
\begin{equation}
\big\|\big[\boldsymbol{\hat{\rho}}_{t^\prime},\mathscr{E}\big(\boldsymbol{\hat{\rho}}_{t^{\prime}}\big)\big]\big\|_{H.S.} = \sqrt{Tr\big\{\big[\boldsymbol{\hat{\rho}}_{t^\prime},\mathscr{E}\big(\boldsymbol{\hat{\rho}}_{t^{\prime}}\big)\big]^{\dagger}\big[\boldsymbol{\hat{\rho}}_{t^\prime},\mathscr{E}\big(\boldsymbol{\hat{\rho}}_{t^{\prime}}\big)\big]\big\}} = 
\end{equation}
\begin{equation}
\sqrt{Tr\big\{\big[\mathscr{E}\big(\boldsymbol{\hat{\rho}}_{t^\prime}\big),\boldsymbol{\hat{\rho}}_{t^{\prime}}\big]\big[\boldsymbol{\hat{\rho}}_{t^\prime},\mathscr{E}\big(\boldsymbol{\hat{\rho}}_{t^{\prime}}\big)\big]\big\}} = \sqrt{2Tr\big\{\mathscr{E}\big(\boldsymbol{\hat{\rho}}_{t^\prime}\big)^{2}\boldsymbol{\hat{\rho}}_{t^{\prime}}^{2}\big\}-2Tr\big\{\big(\boldsymbol{\hat{\rho}}_{t^\prime}\mathscr{E}\big(\boldsymbol{\hat{\rho}}_{t^{\prime}}\big)\big)^{2}\big\}} \leq 
\end{equation}
\begin{equation}
\sqrt{2Tr\big\{\mathscr{E}\big(\boldsymbol{\hat{\rho}}_{t^\prime}\big)^{2}\boldsymbol{\hat{\rho}}_{t^{\prime}}^{2}\big\}} \leq  \sqrt{2Tr\big\{\mathscr{E}\big(\boldsymbol{\hat{\rho}}_{t^\prime}\big)^{2}\big\}Tr\big\{\boldsymbol{\hat{\rho}}_{t^{\prime}}^{2}\big\}} \leq \sqrt{2Tr\big\{\mathscr{E}\big(\boldsymbol{\hat{\rho}}_{t^\prime}\big)\big\}Tr\big\{\boldsymbol{\hat{\rho}}_{t^{\prime}}\big\}} = \sqrt{2}
\end{equation}
where we have employed the fact that $Tr\{\mathbf{\hat{A}}\mathbf{\hat{B}}\}\leq Tr\{\mathbf{\hat{A}}\}Tr\{\mathbf{\hat{B}}\}$ whenever $\mathbf{\hat{A}}$ and $\mathbf{\hat{B}}$ are positive definite trace class operators. 
\end{proof}

\section{Proof of equation (\ref{eqn:exa})}
\label{eqn:apexa}
\begin{equation}
\mathcal{G}_{hs}\big(\mathscr{N}_{t}\circ \Lambda_{t}\big(\boldsymbol{\hat{\rho}}\big)\big) = 
\end{equation}

\begin{equation}
\frac{t}{\sqrt{d^{2}_{tot}-1}}\Big(\|\mathbf{\hat{H}}_{tot}\|_{hs}-\varepsilon Tr\big\{|\{\mathbf{\hat{A}}_{S},\mathbf{\hat{H}}_{S}\}\otimes\sum_{i}E_{i}\big|E_{i}\big\rangle\big\langle E_{i}\big||\big\}\Big)
\end{equation}
\begin{equation}
= \frac{t}{\sqrt{d^{2}_{tot}-1}}\Big(\sqrt{Tr\big\{\mathbf{\hat{H}}_{S}^{2}+\varepsilon\{\mathbf{\hat{A}}_{S},\mathbf{\hat{H}}_{S}\}\otimes\sum_{i}E_{i}\big|E_{i}\big\rangle\big\langle E_{i}
|\big\}+\varepsilon^{2}\mathbf{\hat{A}}_{S}\otimes\sum_{i}E^{2}_{i}|E_{i}\rangle\langle E_{i}|\big\}}
\end{equation}
\begin{equation}
-\sqrt{\varepsilon Tr\big\{|\{\mathbf{\hat{A}}_{S},\mathbf{\hat{H}}_{S}\}\otimes\sum_{i}E_{i}\big|E_{i}\big\rangle\big\langle E_{i}\big||+\varepsilon^{2}\mathbf{\hat{A}}_{S}\otimes\sum_{i}E^{2}_{i}|E_{i}\rangle\langle E_{i}|\big\}}\Big) =
\end{equation}

\begin{equation}
\frac{t}{\sqrt{d^{2}_{tot}-1}}\Big(\sqrt{Tr\big\{\mathbf{\hat{H}}_{S}^{2}\big\}+2\varepsilon N\langle\mathbf{\hat{A}}_{S},\mathbf{\hat{H}}_{S}\rangle_{hs}\langle E_{i}\rangle+\varepsilon^{2}NTr\{\mathbf{\hat{A}}^{2}_{S}\}\langle E^{2}_{i}\rangle}
\end{equation}
\begin{equation}
-\sqrt{2\varepsilon N\langle\mathbf{\hat{A}}_{S},\mathbf{\hat{H}}_{S}\rangle_{hs}\langle E_{i}\rangle+\varepsilon^{2}NTr\{\mathbf{\hat{A}}^{2}_{S}\}\langle E^{2}_{i}\rangle}\Big) =
\end{equation}
\begin{equation}
 \frac{t}{\sqrt{d^{2}_{tot}-1}}\Big(\|\mathbf{\hat{H}}_{S}\|_{hs}+\frac{\varepsilon N\langle\mathbf{\hat{A}}_{S},\mathbf{\hat{H}}_{S}\rangle_{hs}\langle E_{i}\rangle}{\|\mathbf{\hat{H}}_{S}\|_{hs}}-\sqrt{\varepsilon N\langle\mathbf{\hat{A}}_{S},\mathbf{\hat{H}}_{S}\rangle_{hs}\langle E_{i}\rangle}\Big) +\mathcal{O}(\varepsilon^{3/2})  =
\end{equation}
\begin{equation}
 \frac{t}{\sqrt{d^{2}_{tot}-1}}\Big(\|\mathbf{\hat{H}}_{S}\|_{hs}-\sqrt{\varepsilon N\langle\mathbf{\hat{A}}_{S},\mathbf{\hat{H}}_{S}\rangle_{hs}\langle E_{i}\rangle}\Big(1-\frac{\sqrt{\varepsilon N\langle\mathbf{\hat{A}}_{S},\mathbf{\hat{H}}_{S}\rangle_{hs}\langle E_{i}\rangle}}{\|\mathbf{\hat{H}}_{S}\|_{hs}}\Big) +\mathcal{O}(\varepsilon^{3/2}) =  
\end{equation}
\begin{equation}
 \frac{t}{\sqrt{d^{2}_{tot}-1}}\Big(\|\mathbf{\hat{H}}_{S}\|_{hs}(1-\sqrt{\varepsilon}\Omega\big(1-\sqrt{\varepsilon}\Omega\big)\Big) +\mathcal{O}(\varepsilon^{3/2})  
\end{equation}


\begin{thebibliography}{9}

\bibitem{NielGeo}
Nielsen, M.A. 
\textit {"  A geometric approach to quantum circuit lower bounds."}
( Quantum Information
and Computation, Vol 6, Issue 3 Pages 213-262) 
\bibitem{dorth1}
Doherty, A. C., Dowling, M. R., Gu, M., and Nielsen, M. A.
\textit {"Optimal Control,
Geometry, and Quantum Computing.  "}
(Phys. Rev. A 73, 062323(1-7) June 2006) 
\bibitem{dorth2}
Doherty, A. C., Dowling, M. R., Gu, M., and Nielsen, M. A. 
\textit {"Quantum
Computation as Geometry  "}
(Science 311, 1133-1135 February 2006) 

\bibitem{brandt}
Brandt, H.E. 
\textit {" Riemannian geometry of quantum computation.  "}
( Elsevier, Nonlinear
Analysis:Theory, Methods and Applications. Volume 71, Issue 12, 15 December 2009, Pages e474-
e486) 

\bibitem{suss1}
Adam R. Brown and Leonard Susskind
\textit {" Complexity geometry of a single qubit "}
(Phys. Rev. D 100, 046020 – Published 27 August, 2019) 
\bibitem{suss2}
Adam R. Brown
\textit {" A Quantum Complexity Lowerbound from Differential Geometry "}
(Nature Physics 19, 3, 401 (2023)) 
\bibitem{suss3}
Leonard Susskind
\textit {" Black Holes and Complexity Classes "}
(	arXiv:1802.02175 [hep-th]) 

\bibitem{Niel}
Chuang, I.L., Nielsen, M.A. 
\textit {"  Quantum Computation and Quantum Information "}
(Cambridge University Press; Anniversary edition (January 31, 2011)) 


\bibitem{Daniela}
Antonio Falco, Daniela Falco-Pomares, Hermann G. Matthies \textit {"A mathematical model for a universal digital quantum computer with an application to the Grover-Rudolph algorithm"}(arXiv:2503.13388 (2025) )

\bibitem{grom}
M. Gromov \emph{ "Structures métriques pour les variétés Riemanniennes" (Textes Mathématiques, vol. 1, Paris, 1981, Edited by J. Lafontaine and P. Pansu)}





\bibitem{Bu}
Kaifeng Bu, Roy J.Garcia, Arthur Jaffe, Dax Enshan Koh, And Lu Li\emph{ "Complexity of Quantum Circuits Via Sensitivity, Magic, And Coherence" ( Commun. Math. Phys. 405, 161 (2024).}

\bibitem{Schloss}
 M. Schlosshauer
  \textit{ "Decoherence and the Quantum-To-Classical Transition"}
  (Springer-Verlag, Berlin Heidelberg,
  2007)


\bibitem{cone}
G. Dirr, U. Helmke, I. Kurniawan, T. Schulte-Herbrueggen \emph{ "Lie-Semigroup Structures for Reachability and Control of Open Quantum Systems: Viewing Markovian Quantum Channels as Lie Semigroups and GKS-Lindblad Generators as Lie Wedge" (Reports on Mathematical Physics
Volume 64, Issues 1–2, August–October 2009, Pages 93-121)}




\bibitem{annealing}
Diego de Falco, Dario Tamascelli
\textit {" An introduction to quantum annealing  "}
(RAIRO-Theor. Inf. Appl. 45 (2011) 99-116) 

\bibitem{cz}
Rufus Lawrence, Ales Wodecki, Johannes Aspman, and Jakub Marecek \textit {"On the Random Schr¨ odinger Equation and Geometric Quantum Control R"} (arXiv:2503.04617, (2025))


\bibitem{laflame}
Emanuel Knill and Raymond Laflamme \textit {"Theory of quantum error-correcting codes"} (Phys. Rev. A 55, 900 – Published 1 February, 1997 )




\bibitem{Jeff}
Jefferson, R.A., Myers, R.C.
\textit {"  Circuit complexity in quantum field theory  "}
( J. High Energ. Phys. 2017, 107 (2017). https://doi.org/10.1007/JHEP10(2017)107) 

\bibitem{cot}
Jordan Cotler, Nicholas Hunter-Jones, Junyu Liu Beni Yoshida
\textit {"  Chaos, complexity, and random matrices "}
(Regular Article - Theoretical Physics
Open access
Published: 09 November 2017
Volume 2017, article number 48, (2017)) 


\bibitem{B2}
Adam R. Brown, Leonard Susskind, and Ying Zhao
\textit {"  Quantum complexity and negative curvature"}
(Phys. Rev. D 95, 045010 – Published 22 February, 2017) 



\bibitem{gour}
Eric Chitambar and Gilad Gour
\textit {" Quantum resource theories "}
( Rev. Mod. Phys. 91, 025001 – Published 4 April, 2019) 

















\end{thebibliography}
\end{document}